\documentclass{article}

\usepackage[preprint]{neurips_2026}

\usepackage[utf8]{inputenc} %
\usepackage[T1]{fontenc}    %
\usepackage{hyperref}       %
\usepackage{url}            %
\usepackage{booktabs}       %
\usepackage{amsfonts}       %
\usepackage{nicefrac}       %
\usepackage{microtype}      %
\usepackage{xcolor}         %
\usepackage{listings}
\usepackage{xcolor}
\usepackage{graphicx}
\usepackage{booktabs}
\usepackage{multirow}
\usepackage{tcolorbox}
\usepackage{listings}
\usepackage{microtype}
\usepackage{caption}
\usepackage{xcolor}
\usepackage{xspace}
\usepackage{amsmath,amssymb}%
\usepackage{amsthm}
\usepackage{mathtools}
\usepackage{bm}

\newtheorem{assumption}{Assumption}
\newtheorem{condition}{Condition}
\newtheorem{remark}{Remark}
\newtheorem{proposition}{Proposition}

\usepackage{wrapfig}
\usepackage{inconsolata}
\usepackage[T1]{fontenc}
\usepackage{microtype}
\usepackage{tikz, xcolor}
\usetikzlibrary{positioning, arrows.meta}
\usepackage{minted}
\usepackage{subcaption}
\newcommand{\methodname}{\textsc{MaskGen}\xspace}

\newcommand{\eg}{\emph{e.g.}\xspace}
\newcommand{\ie}{\emph{i.e.}\xspace}

\usepackage[utf8]{inputenc}
\usepackage{amssymb}

\definecolor{addedFg}{HTML}{1A7F37}
\definecolor{diffgreen}{rgb}{0.88, 1.0, 0.88}

\newcommand{\subpara}[1]{\textbf{#1}\xspace}

\title{Why Invariance is Not Enough for Biomedical Domain Generalization and How to Fix It}

\author{%
  Sebo Diaz \\
  MIT \\
  \texttt{sdd@mit.edu} \\
  \And
  Polina Golland \\
  MIT \\
  \texttt{polina@mit.edu} \\
  \And
  Elfar Adalsteinsson \\
  MIT \\
  \texttt{elfar@mit.edu} \\
  \And
  Neel Dey \\
  MIT, MGH, HMS \\
  \texttt{ndey@mgh.harvard.edu} \\
}

\begin{document}

\maketitle

\begin{figure}[h]
    \centering
    \includegraphics[width=0.975\linewidth]{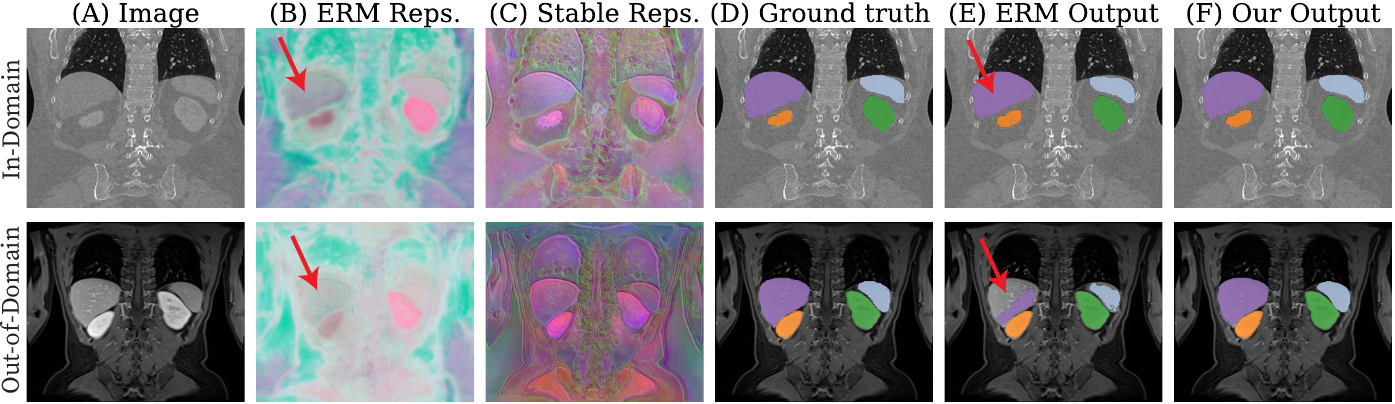}
    \caption{\small{
    \textbf{Training on Stable Representations.} 
    When standard ERM models are trained on one domain (A, top) but tested on another (A, bottom), their representations are unstable under domain shifts (B).
    While standard training is performant in-domain (E, top), this instability significantly degrades performance on out-of-distribution images (E, bottom). \methodname provides a theoretically grounded and easy-to-implement framework that exploits \textit{both} in-domain image intensities (A, top) and stable, domain-invariant representations extracted by foundation models (C, top) to train robust models that effectively generalize to new domains automatically without any adaptation (F, bottom). Representations visualized by mapping 3 arbitrary channels to RGB.
    }}
    \label{fig:placeholder}
\end{figure}

\begin{abstract}
    We present \methodname, a theoretically grounded and deliberately simple approach for domain generalization in 3D biomedical image segmentation. Modern segmentation models degrade sharply under shifts in modality, disease severity, clinical sites, and more, limiting their reliable adoption. Existing generalization methods address this using extreme augmentations, hand-engineered domain statistics mixing, or architectural redesigns that add significant implementation overhead while yielding inconsistent performance across biomedical settings. \methodname instead presents a principled learning strategy with marginal overhead that utilizes both source-domain image intensities and domain-stable foundation model representations to train robust segmentation models. As a result, \methodname achieves strong gains in both fully supervised and few-shot segmentation across broad clinical shifts in biomedical studies. Unlike prior approaches, \methodname is architecture- and loss-agnostic, compatible with standard augmentation pipelines, easy to implement, and tackles arbitrary anatomical regions. Its implementation is freely available at \href{https://github.com/sebodiaz/MaskGen}{https://github.com/sebodiaz/MaskGen}.
\end{abstract}

\section{Introduction}
Medical image segmentation is considered largely solved in controlled environments: with a few hundred labeled examples, strong augmentation, and a good training recipe, models achieve high performance when the test data matches the training domain~\citep{isensee2021nnu,isensee2024nnu}. Unfortunately, real-world deployment is far less forgiving. Annotated medical datasets are often small, and domain shifts in patient populations, imaging sites, acquisition parameters, and modalities all significantly degrade performance, forcing costly, time-consuming cycles of retraining and additional manual annotation.

\looseness=-1
Domain generalization (DG) methods address these shifts during training, removing the need for retraining on distribution-shifted data. They typically operate at the data- or model-level. Data-level DG expands training diversity through extreme augmentations~\cite{ouyang2022causality} or style transfer~\cite{zhou2021domain}, but cannot guarantee coverage of unseen target domains and may conflict with standard augmentation pipelines. Model-level DG mixes feature statistics across domains~\cite{zhou2021domain,zhou2022generalizable}, but is computationally expensive and assumes that target data lies within the span of training domains. Semi- and unsupervised domain adaptation methods sidestep these assumptions by finetuning or training image translation models on each new target domain at test time~\cite{chen2019synergistic,chen2022contrastive,karani2021test, wang2020tent}, but demand machine learning expertise, infrastructure, and repeated effort for each new domain, limiting their practicality for biomedical users.

\looseness=-1
Recent foundation models for interactive biomedical segmentation perform well by training on large, diverse datasets using prompts such as clicks, scribbles, and bounding boxes~\cite{he2024vista3d,isensee2025nninteractive,wong2025multiverseg,wong2024scribbleprompt}. However, they do not scale to automated workflows: prompting thousands of images is clinically infeasible, and these models often struggle on underrepresented domains even with user input. In contrast, anatomix~\cite{dey2024learning} is a synthetically trained 3D biomedical foundation model that extracts task-agnostic \textit{representations} that are stable under nuisance imaging variations (\eg, scanner types, modalities, acquisition protocols). While intuition suggests that training a model on these domain-invariant representations (instead of raw image intensities) could enable DG, we find that this strategy does not outperform just using the image as is when trained with extensive data augmentation (Sec. \ref{subsec:analysis_and_ablations} and Fig. \ref{fig:dropout_sensitivity}).

This performance and usability gap in the DG literature motivates the need for combining stable, domain-invariant representations with in-domain image intensities to train robust predictors that generalize. Prior test-time-adaptation work~\cite{eastwood2023spuriosity,anti-causal} proposes methods that combine stable and unstable representations to improve robustness. However, these approaches assume access to test data by design and require finetuning or adaptation on each task, which, again, cannot be practically executed in the clinic. In this paper, we develop an algorithm that trains robust segmentors combining in-domain and domain-invariant representations, \textit{without} needing any test-time adaptation.

\begin{figure}[t]
\centering
\begin{subfigure}{0.67\linewidth}
\centering
\begin{minted}[
  highlightlines={1,5-8},          %
  highlightcolor=diffgreen,         %
  linenos,
  numbersep=2pt,
  breaklines,
  fontsize=\footnotesize, %
]{python}
feature_ex = load_feature_extractor(trainable=False) # f_phi
train_net = initialize_model(trainable=True) # h_theta
for (images, labels) in dataloader: # (x_u, Y)
    optimizer.zero_grad()
    with torch.no_grad():
        features = feature_ex(images) # x_s
        inputs = torch.cat([images, features], dim=1) # z
        inputs = F.dropout3d(inputs, p=drop_prob) # \tilde{z}
    outputs = train_net(inputs)
    loss = loss_fn(outputs, labels)
    loss.backward()
    optimizer.step()
\end{minted}
\label{fig:train_loop_diff}
\end{subfigure}
\hspace{0.02\linewidth}
\begin{subfigure}{0.3\linewidth}
\centering
\begin{tikzpicture}[
    scale=1.2,
    transform shape,
    node distance=1.2cm and 0.8cm,
    >={Stealth[length=4pt, width=3pt]},
    every node/.style={font=\scriptsize},
    latent/.style={
        circle, draw, minimum size=20pt, 
        inner sep=0pt, font=\large
    },
]
    \definecolor{stableblue}{HTML}{2563EB}
    \definecolor{spuriousred}{HTML}{DC2626}
    \definecolor{labelyellow}{HTML}{BA8E23}
    \definecolor{envgray}{HTML}{4B5563}
    \node[latent, draw=labelyellow, thick] (Y) {$Y$};
    \node[latent, draw=envgray, thick, right=of Y] (E) {$E$};
    \node[latent, draw=stableblue, thick, below=of Y] (Xs) {$X_s$};
    \node[latent, draw=spuriousred, thick, below=of E] (Xu) {$X_u$};
    \draw[->, thick, stableblue] (Y) -- (Xs)
        node[midway, sloped, above, font=\tiny, text=stableblue] {invariant};
    \draw[->, thick] (Y) -- (Xu);
    \draw[->, thick, spuriousred] (E) -- (Xu)
        node[midway, sloped, above, font=\tiny, text=spuriousred] {shifts};
    \node[below=0.5pt of Xs, font=\tiny, text=stableblue] {stable};
    \node[below=0.5pt of Xu, font=\tiny, text=spuriousred] {unstable};
    \node[above=1pt of E, font=\tiny, text=envgray] {environment};
    \node[above=1pt of Y, font=\tiny, text=labelyellow] {label};
\end{tikzpicture}
\label{fig:pgm}
\end{subfigure}
\hspace{-4em}
\caption{\textbf{Method.} \textbf{Left}: \methodname is exceedingly simple to implement. Given a standard PyTorch training loop, the \textcolor{addedFg}{green} lines are the only additions required. \textbf{Right}: The probabilistic graphical model we assume for domain generalization. Label $Y$ generates both stable $X_s$ and unstable $X_u$ variables and the environment $E$ influences only $X_u$.} %
\label{fig:method_overview}
\end{figure}

\looseness=-1
\subpara{Contributions.} We present \methodname, a deliberately simple and effective domain generalization method for arbitrary domain shifts with minimal overhead. \methodname trains robust segmentation models that jointly exploit \textit{both} in-domain image intensities and stable, domain-invariant representations extracted by frozen, non-trainable foundation models~\cite{dey2024learning}. We provide a theoretical and empirical analysis of stable and unstable feature combination under distribution shift and find that naïve combination induces shortcut learning, where models over-rely on in-domain unstable features and fail to generalize. To achieve the best of both stability and instability, we propose a simple, principled regularization mechanism to prevent this shortcut and force the model to jointly exploit both information sources. Experimentally, \methodname achieves excellent performance across multiple biomedical contexts and data availability settings, enabling reliable segmentation when applied to new domains. Further, while \methodname does not use test-time adaptation, it performs competitively with methods that do retrain on (unlabeled) target-domain data, greatly enhancing its practicality for biomedical users.

\section{Related Work}
\label{sec:related_work}
\noindent\textbf{Domain Generalization (DG).} DG methods aim to learn domain-invariant segmentors \underline{before} deployment. Data-level DG methods often use strong data augmentation to expand the training domain, such as convolving images with randomly initialized filters~\cite{ouyang2022causality,xu2020robust}, stitching samples~\cite{zhang2017mixup}, and incorporating spatial mixing~\cite{ouyang2022causality}, among others~\cite{devries2017improved}. Other model-level DG approaches modify architectures or losses to reduce sensitivity to spurious cues~\cite{huang2020self,zhou2021domain,zhou2022generalizable}. When domain labels are available (\eg, imaging sites or sequences), DG methods can explicitly incorporate domain differences into their modeling or architectures~\cite{hu2022domain,liu2021feddg,tiwary2025langdaug,zhou2022generalizable}. However, DG methods often add non-trivial computational overhead, do not reliably transfer across applications, may conflict with standard training pipelines, and explicit domain labels are rare in medical imaging datasets. \methodname avoids this dependency on domain labels, requires marginal extra computation, and is fully generic w.r.t. target applications, improving both the practicality and deployment flexibility of DG segmentation.

\subpara{Unsupervised Domain Adaptation (UDA).} UDA methods address domain shifts by adapting models \textit{at} deployment using unlabeled test data. Many approaches optimize surrogate losses that encourage confident or consistent predictions, leveraging auxiliary tasks~\cite{sun2020test}, entropy minimization~\cite{wang2020tent}, batch statistics updates~\cite{valanarasu2024fly,dong2024medical}, augmentation consistency~\cite{weihsbach2025dg}, or contrastive regularization~\cite{chen2022contrastive}. Others focus on aligning feature distributions, either statistically~\cite{chen2019synergistic,kang2019contrastive,karani2021test,sun2016deep,zheng2024dual} or via adversarial training~\cite{tzeng2017adversarial}. Pseudo-labeling methods similarly use confident predictions as weak supervision~\cite{wu2024fpl+, zhang2024iplc, Zhang_2024_CVPR, Zhao_2024_WACV, zou2019confidence}, but often fail beyond their intended domain shifts and anatomical regions and require considerable machine learning resources to deploy. In contrast, \methodname does not require test-time adaptation and generically handles arbitrary domain shifts and applications.

\subpara{Promptable Foundation Models (PFMs).} By training on broad aggregated datasets, recent large-scale promptable models~\cite{kirillov2023segment,ravi2024sam} generate segmentation masks conditioned on user inputs such as clicks, bounding boxes, or text, enabling flexible segmentation without task-specific retraining~\cite{butoi2023universeg,hoopes2025voxelprompt,isensee2025nninteractive,wong2024scribbleprompt}. However, most PFMs operate on 2D slices and require per-image prompting, limiting their practicality for large studies involving thousands of images. Further, they are unsuitable for general-purpose feature extraction as they can underperform on distribution-shifted data. %

\looseness=-1
\subpara{Exploiting Stable Representations.} Prior work in non-medical settings has explored obtaining and leveraging stable, domain-invariant representations by solving a constrained optimization problem to minimize risk across all training domains \cite{arjovsky2019invariant}. However, these representations can underperform in practice, as optimizing for cross-domain stability discards valuable domain-specific information~\cite{eastwood2023spuriosity, rosenfeld2020risks, shui2022benefits, zhao2022fundamental}.The challenge, then, is not choosing between invariant and domain-specific information but effectively utilizing both. Recent work trains two separate models, one capturing stable and another capturing unstable representations, and combines the models at test time~\cite{eastwood2023spuriosity, anti-causal}. %
However, these two-stage, two-model pipelines are both computationally expensive and may require careful domain-specific tuning at test time. \methodname instead learns a \textit{single} predictor at training time alone.

\section{Methods and Analysis}
\label{sec:analys_and_methods}
\noindent\textbf{Preliminaries.} Let $\mathcal{X}$ and $\mathcal{Y}$ be the input and label spaces, respectively. Consider data generated from multiple environments, where each environment $e \in \mathcal{E}$ (\eg, scanner, modality) induces a joint distribution $\mathbb{P}_{e}$ over $\mathcal{X}\times\mathcal{Y}$, with training environments $\mathcal{E}_{train} \subseteq \mathcal{E}$. We assume that environment labels of $\mathcal{E}_{train}$ are not available during training, as clinical datasets often lack reliable or relevant metadata. A predictor $h_{\theta}$ observes only pooled data $D_{\mathrm{train}} = \{(x_i, y_i)\}_{i=1}^N$, where each latent environment $e_i \in \mathcal{E}_{\mathrm{train}}$ is unobserved. Empirical Risk Minimization (ERM)~\cite{vapnik1999overview} optimizes $\theta^* = \arg\min_{\theta} \mathbb{E}_{(x,y) \sim \bar{\mathbb{P}}}\bigl[\ell(h_{\theta}(x),y)\bigr],$ where $\ell$ is the loss and $\bar{\mathbb{P}}$ is a mixture of training environments. 
We aim to learn $h_{\theta}$ that generalizes to unseen environments $\mathcal{E}_{\mathrm{test}} \subseteq \mathcal{E} \setminus \mathcal{E}_{\mathrm{train}}$ without access to any subset of $\mathcal{E}$ at train or test time. We state the algorithm in Sec.~\ref{subsec:method} and provide analysis in Sec.~\ref{subsec:analysis}.

\subsection{Methods} \label{subsec:method}
\methodname is outlined in Figure~\ref{fig:method_overview} alongside our graphical model assumptions. It is  comprised of three key components: domain-invariant representation extraction, combination with domain-specific image intensities, and regularizing the combination with view masking to prevent shortcuts.

\subpara{Stable Representation Extractor.}
Let $f_{\phi}: \mathcal{X} \rightarrow \mathcal{S}$ be a representation extractor parameterized by $\phi$, mapping inputs $x \in \mathbb{R}^{1\times H\times W \times D}$ to representations $s = f_{\phi}(x) \in \mathbb{R}^{d\times H\times W \times D}$ where $d$ is the embedding dimension and $H$, $W$, and $D$ indicate spatial dimensions. Modern foundation models obtain transferable representations that can be used in the domain generalization setting~\cite{dey2024learning, isensee2025nninteractive}. In particular, the extractor we use in our experiments produces approximately invariant (stable) representations~\cite{dey2024learning}, \ie, $f_{\phi}(x_e) \approx f_{\phi}(x_{e'}) \ \forall \ e,e'\in \mathcal{E}$, as quantitatively verified in Appendix~\ref{app:representation_analysis}.

\subpara{Combining Information Sources.}
At test time, prior work combines $\mathcal{S}$ and $\mathcal{X}$ by merging logits from models trained separately on $\mathcal{S}$ and $\mathcal{X}$ or by using a model trained on $\mathcal{S}$ as a pseudo-labeler to finetune a model trained on $\mathcal{X}$~\cite{eastwood2023spuriosity, anti-causal}. To avoid test-time adaptation or training separate models for each source of information, we instead provide the predictor with an input that contains both the raw image and the representation extractor's outputs: $z = \text{concat}(x,\, f_{\phi}(x))
      \in \mathbb{R}^{(1+d) \times H\times W \times D}.
$
The predictor $h_{\theta}: \mathcal{Z} \rightarrow \mathcal{Y}$ then operates on $z$ directly. We write $x_u := x$ for the \emph{unstable} source and $x_s := f_{\phi}(x)$ for the \emph{stable} source, such that $z = (x_u, x_s)$. In the analysis that follows, we use $X_u, X_s,$ and $Z$ to denote the corresponding random variables.

\looseness=-1
\subpara{Regularizing Feature Combination.} Directly training $h_{\theta}$ on $z = (x_u, x_s)$ using ERM does not produce generalization. As $x_u$ has in-domain signal, the predictor trained with an in-domain loss (ERM) is incentivized to rely on $x_u$ to minimize the ERM loss and ignore the domain-invariant features $x_s$. To prevent this, we randomly mask different channels (features) of $z$ during training to produce random partial \textit{views} of $z$. As $h_{\theta}$ is now trained on only partial views that may or may not include $x_u$, it cannot ignore the stable inputs $x_s$ and must use them, as shown by our theoretical analysis below. In practice, we stochastically mask each of the $(1+d)$ channels in $z$ at a rate $r$.

\subsection{Analysis} \label{subsec:analysis}
While masking views of $z$ might intuitively prevent over-reliance on either information source, does it structurally prevent over-reliance on $X_u$? Further, is this added complexity over training on the domain-invariant $X_s$ alone justified? Here, we investigate the following questions:
\begin{enumerate}
    \item Does masking views of $z$ theoretically force $h_{\theta}$ to use $X_{s}$? (Proposition~\ref{prop:stationarity})
    \item Does training on $(X_u, X_s)$ improve over training on $X_s$ alone? (Proposition~\ref{prop:ceiling})
\end{enumerate}

Below, $H(\cdot)$ denotes entropy, $H(\cdot \mid \cdot)$ conditional entropy, $H_e(\cdot)$ and $H_e(\cdot \mid \cdot)$ denotes entropy and conditional entropy under a particular environment $e$, $I(\cdot \,;\, \cdot )$ mutual information, $\star$ is cross-correlation, $\ell$ is a proper scoring rule, and $\sigma$ is a pointwise non-linearity.
As $h_{\theta}$ takes the concatenated input $(X_u, X_s)$, its first layer admits a partition of weights into slices $W_u$ and $W_s$ acting on the unstable and stable inputs, respectively. 
This holds when the first layer is linear in the input, as in CNNs (where $W_u,W_s$ are kernel slices) and ViTs (where $W_u,W_s$ are slices of the shared embedding matrix). We assume the following (with the first 3 assumptions following~\cite{eastwood2023spuriosity}):
\begin{assumption}[Stable inputs]
\label{ass:stable}
$Y \perp E \mid X_s$, \ie, the stable inputs encode information about $Y$ that is environment-invariant.
\end{assumption}
\begin{assumption}[Unstable inputs]
\label{ass:unstable}
$Y \not\perp E \mid X_u$, \ie, the relationship between the in-domain input $X_u$ and $Y$ depends on the environment.
\end{assumption}
\begin{assumption}[Complementary information]
\label{ass:complementarity}
$Y \not\perp_e X_u \mid X_s$ for every $e \in \mathcal{E}$, where $\perp_e$ denotes conditional independence under $\mathbb{P}_{e}$,
\ie, in each environment $X_u$ carries task-relevant 
information beyond what $X_s$ provides.
\end{assumption}
\begin{assumption}[Informative stable inputs]
\label{ass:informative}
$I(Y;\, X_s) > 0$, \ie, $X_s$ possesses nonzero predictive information about $Y$.
\end{assumption}
Assumptions~\ref{ass:stable}--\ref{ass:complementarity} are explicit in the graphical model (Fig.~\ref{fig:method_overview}, right). Assumption~\ref{ass:informative} is to preclude trivially domain-stable inputs (\eg, zero-valued features no matter the environment indexed).

\begin{remark}[Shortcut Learning]
\label{rem:temptation}
 Assumptions~\ref{ass:unstable} and~\ref{ass:complementarity} imply that $X_u$ carries in-domain predictive signal for $Y$, but that this relationship shifts across environments. 
 Given both $(X_u, X_s)$, the model is incentivized to primarily use $X_u$ due to its predictive in-domain signal when training with an in-domain ERM loss, thereby ignoring $X_s$ and failing on test data when the environment shifts.
\end{remark}

We hypothesize that shortcut learning occurs not because the stable inputs are uninformative across environments but due to the training objective focusing only on the training environments, which is empirically validated in Section~\ref{sec:experiments}. Below, we show that view masking as a feature combiner is a principled strategy to regularize the joint use of $X_u$ and $X_s$ and mitigate these harmful effects.

\subpara{View Masking as a Feature Combiner.}
Here, we make the simplifying assumption that random view masking either masks all feature channels of $f_{\phi}$ or just the original input $x$. Let $M = (M_u, M_s) \in \{0,1\}^{2} \setminus \{(0,0)\}$, be a random view mask sampled independently of $(x,y)$. Writing the masked input as $\tilde{z} := (M_u\, x_u,\; M_s\, x_s)$, the training objective is
\(
\mathcal{L}(\theta)
= \mathbb{E}_{(x,y)\sim\bar{\mathbb{P}}}\;
      \mathbb{E}_{M}\!
      \bigl[\, \ell\bigl(h_\theta(\tilde{z}),\, y\bigr) \bigr].
\)
Let $\mathcal{R}_{\mu}(\theta) := \mathbb{E}_{(x,y)\sim \bar{\mathbb{P}}}\bigl[\ell\bigl(h_\theta(\mu \odot z),\, y\bigr)\bigr]$ denote the risk under a fixed mask $\mu \in \{(1,1),\,(1,0),\,(0,1)\}$, and let
$\pi_\mu := \mathbb{P}(M\!=\!\mu)$.  Then, the training objective additively decomposes as,
\begin{equation}
    \mathcal{L}(\theta)
    = \sum_{\mu}
      \pi_\mu \, \mathcal{R}_\mu(\theta),
    \qquad
    \mu \in \{(1,1),(1,0),(0,1)\}.
    \label{eq:risk-decomp}
\end{equation}
which is a convex combination of the joint, unstable, and stable risks, respectively. Note that, with view masking, $\mathcal{R}_{(0,1)}$ penalizes solutions that ignore $x_s$, directly counteracting shortcut learning. %

\begin{condition}[Non-conflicting interaction]
\label{cond:nai}
For any $\theta^*$ where $h_{\theta^*}$ ignores $X_s$, the joint risk $\mathcal{R}_{(1,1)}$ gradient w.r.t. the stable channel kernel slices $W_s$ does not oppose the corresponding gradient of the stable-only risk $\mathcal{R}_{(0,1)}$:
\begin{equation}
\Bigl\langle
        \nabla_{W_s} \mathcal{R}_{(1,1)}(\theta^*),\;
        \nabla_{W_s} \mathcal{R}_{(0,1)}(\theta^*)
    \Bigr\rangle
    \;\geq\; 0.
\end{equation}
\end{condition}

\begin{remark}
The condition is mild. At any $\theta^*$ ignoring $X_s$, the first-layer activation under mask $(1,1)$ reduces to $\sigma(W_u \star X_u)$, identical to the $(1,0)$ regime. Activating $W_s$ introduces an additive signal $W_s \star X_s$ inside the nonlinearity that carries information about $Y$ (Assumption~\ref{ass:informative}) and is environment-invariant (Assumption~\ref{ass:stable}). Violating this condition would implausibly require that this additional informative, domain-stable signal actively degrades prediction and the model is harmed by using strictly more task-relevant information. This condition is also empirically verified in Appendix~\ref{app:condition_1}. %
\end{remark}

\begin{proposition}[Stationarity forces use of stable inputs]
\label{prop:stationarity}
Given Assumptions~\ref{ass:stable}--\ref{ass:informative}, let $h_{\theta}$ be a model whose first layer computes: $a^{(1)} = \sigma( W_{u} \star X_{u} + W_{s} \star X_{s} )$
where $W_u, W_s$ denote the first-layer kernel slices corresponding to the unstable and stable input channels, respectively, and are differentiable w.r.t. $\theta$. Suppose the function class $\{h_{\theta}(\mathbf{0}, \cdot) : \theta \in \Theta \}$ can realize the Bayes-optimal predictor for $(X_s, Y)$ and that $\mathcal{R}_{(0,1)}$ has no spurious stationary points w.r.t. $W_s$. Given Condition~\ref{cond:nai}, if $\pi_{(0,1)} > 0$, then at no stationary point $\theta^*$ of $\mathcal{L}$ can $h_{\theta^*}$ ignore $X_{s}$
\textbf{(Proof provided in Appendix~\ref{app:proofs})}.%
\end{proposition}

Proposition~\ref{prop:stationarity} shows that view masking forces the predictor to use the stable inputs. We now turn to the natural follow-up question: why not train on the stable inputs $X_s$ alone? If $X_s$ is environment-invariant (Assumption~\ref{ass:stable}), discarding $X_u$ should eliminate domain shift entirely. The answer is that invariance alone does not guarantee sufficient performance, as shown below.

\begin{proposition}[Stable-only performance ceiling]
\label{prop:ceiling}
Under Assumptions~\ref{ass:stable}--\ref{ass:informative}, $H(Y | X_s)$ is environment-invariant, and the Bayes-optimal stable-only predictor $h_s^*(x_s) := \mathbb{E}[Y | X_s = x_s]$ achieves this value as its risk in every environment. Therefore, in any environment $e$, the Bayes-optimal joint risk $H_e(Y | X_s, X_u) < H(Y | X_s)$, with strict inequality, where $H_e(\cdot)$ denotes entropy under $\mathbb{P}_e$ \textbf{(Proof provided in Appendix~\ref{app:proofs})}.
\end{proposition}

\begin{remark}
    Proposition~\ref{prop:ceiling} motivates using both $X_u$ and $X_s$, as discarding $X_u$ imposes an upper bound on performance. In parallel, Proposition~\ref{prop:stationarity} shows that $h_{\theta}$ cannot become solely dependent on $X_u$. Together, they imply 
    that $X_s$ provides a floor on performance and additionally providing $X_u$ improves it in any given environment. 
\end{remark}

This analysis reflects a fundamental limitation of invariant approaches, also supported by previous work~\cite[Theorem 5.1]{rosenfeld2020risks}. %
Therefore, we do not discard the in-domain features $X_u$ but exploit their transferable information, while ensuring the predictor does not learn shortcuts that ignore $X_s$. %

\subpara{Implementation Details.}
Fig.~\ref{fig:method_overview} shows \methodname only needs five additional lines of code over a standard training loop. We implement $h_{\theta}$ as a 3D UNet, which remains competitive with transformer-based alternatives~\cite{isensee2024nnu}. The model has $36.9$M parameters with an initial embedding channel dimension of $20$, doubled at every resolution across $5$ downsampling operations. For $f_{\phi}$, we use anatomix\cite{dey2024learning}, which produces per-voxel embeddings and has $5.9$M parameters. As anatomix was trained entirely on synthetic data, it eliminates data leakage risks in our experiments. \textbf{$f_{\phi}$ is frozen during training and is not finetuned} and we empirically verify its stability to domain shifts in Appendix~\ref{app:representation_analysis}. 
We train with a Dice and cross-entropy loss with all other training details in Appendix~\ref{app:implementation}.

\begin{figure}[!t]
  \centering
  \includegraphics[width=0.925\textwidth]{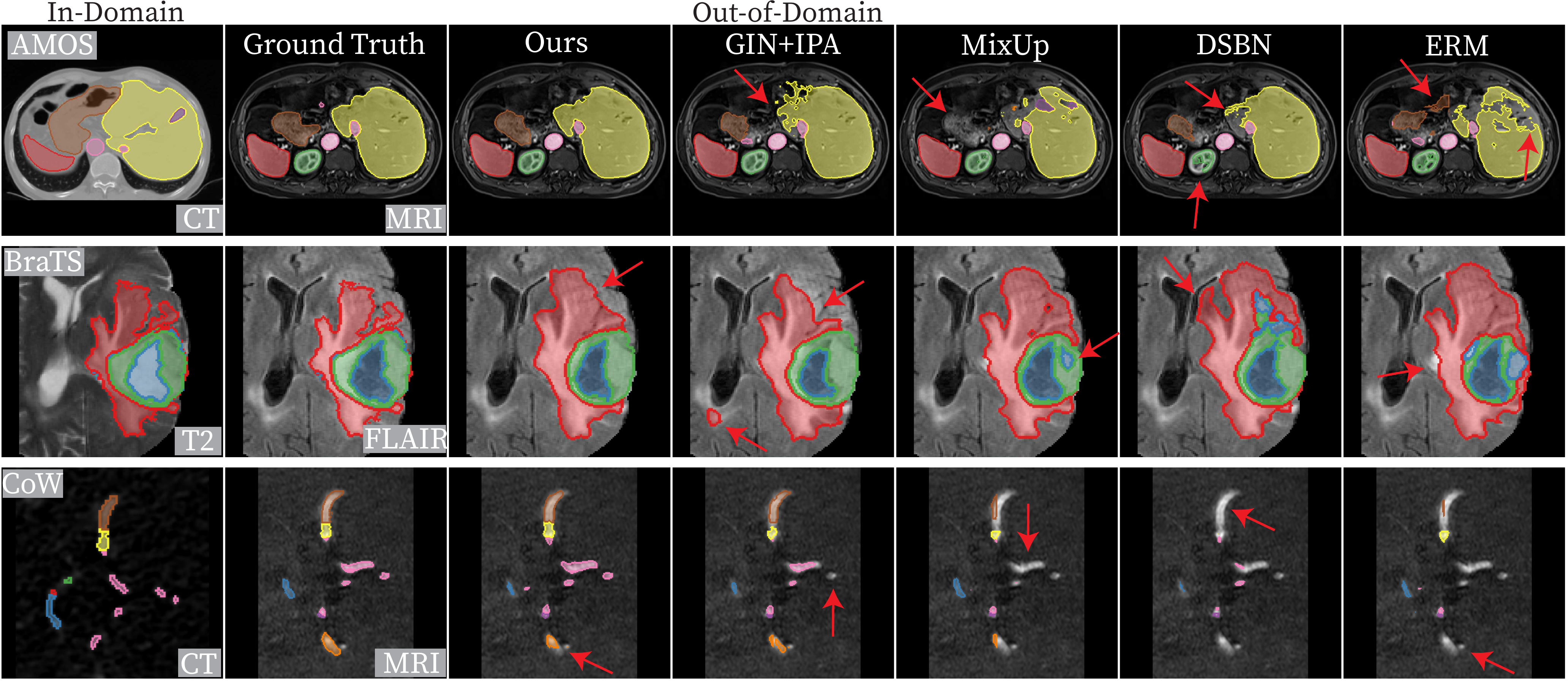}
  \caption{\textbf{Qualitative all-data segmentation results.} Rows correspond to different datasets and domain shift types. Cols. 1 and 2 visualize in-domain training and out-of-domain testing samples, respectively. Cols. 3--7 illustrate qualitative results on out-of-domain test set examples. %
  } 
  \label{fig:qualitative_250ks}
\end{figure}

\begin{table}[!t]
\caption{\textbf{All-data domain generalization results.} Mean Dice score (w/ std. error) across datasets against domain generalization baselines, with avg. ranking on the far-right.}
\centering
\footnotesize
\setlength{\tabcolsep}{2.5pt}
\begin{tabular}{l c c c c c c c}
\toprule
\textbf{Method} &
\begin{tabular}{c} \textbf{AMOS}\cite{ji2022amos}\\ \footnotesize Modality \end{tabular} &
\begin{tabular}{c} \textbf{BraTS}\cite{antonelli2022medical}\\ \footnotesize Sequence \end{tabular} &
\begin{tabular}{c} \textbf{CoW}\cite{topcowchallenge}\\ \footnotesize Modality \end{tabular} &
\begin{tabular}{c} \textbf{HVSMR}\cite{pace2024hvsmr}\\ \footnotesize Disease \end{tabular} &
\begin{tabular}{c} \textbf{PanDG}\cite{zhang2025rethink}\\ \footnotesize Phase \end{tabular} &
\begin{tabular}{c} \textbf{Prostate}\cite{liu2021feddg}\\ \footnotesize Site \end{tabular} &
\textbf{\begin{tabular}{c} Avg.\\ Rank \end{tabular}} \\
\midrule
ERM & 57.4 (1.7) & 38.9 (1.8) & 24.6 (1.5) & 63.6 (3.0) & 25.3 (2.5) & 83.0 (0.7) & 8.0 \\
GIN\cite{ouyang2022causality} & 65.1 (0.7) & 29.9 (1.6) & 17.9 (1.5) & 64.1 (3.0) & 41.1 (1.8) & 82.1 (0.8) & 7.2 \\
GIN+IPA\cite{ouyang2022causality} & $\underline{66.5 (0.8)}$ & $\underline{44.7 (1.6)}$ & $\underline{37.6 (1.9)}$ & 64.7 (2.6) & $\textbf{50.3 (1.4)}$ & $\underline{84.1 (0.6)}$ & $\underline{2.5}$ \\ 
MixUp\cite{zhang2017mixup} & 61.5 (0.9) & 41.2 (1.7) & 36.3 (1.5) & 64.5 (3.2) & 26.7 (2.5) & 82.5 (0.7) & 5.5 \\
CutOut\cite{devries2017improved} & 55.5 (1.7) & 41.8 (1.6) & 26.5 (1.9) & 64.9 (3.5) & 30.5 (2.4) & 84.0 (0.6) & 5.0 \\
DSBN\cite{zhou2022generalizable} & 58.0 (1.4) & 40.5 (1.7) & 29.6 (1.9) & \underline{$66.1 (3.2)$} & 26.3 (2.6) & 83.4 (0.6) & 5.2 \\
MixStyle\cite{zhou2021domain} & 53.9 (1.7) & 40.7 (1.6) & 27.9 (1.4) & 64.4 (2.6) & 24.0 (2.5) & 82.3 (0.7) & 7.7 \\
RSC\cite{huang2020self}      & 60.8 (1.3) & 42.1 (1.6) & 26.0 (1.9) & $\textbf{66.2 (2.9)}$ & 28.2 (2.4) & 83.4 (0.6) & 4.3 \\
nnUNet\cite{isensee2021nnu}   & 48.0 (2.7) & 35.5 (1.8) & 13.7 (2.0) & 65.0 (2.5) & 30.1 (2.7) & 71.6 (3.6) & 8.0  \\
\midrule
Ours & $\textbf{67.5 (0.9)}$ & $\textbf{46.2 (1.6)}$ & $\textbf{47.8 (1.3)}$ & 65.8 (3.4) & $\underline{46.0  (1.9)}$ & $\textbf{84.7 (0.5)}$ & $\textbf{1.5}$ \\
\bottomrule
\end{tabular}
\label{tab:main_results_full}
\end{table}

\section{Experiments}
\label{sec:experiments}
\subsection{Experimental Setup}
\noindent\textbf{Datasets and Domain Shifts.} We use 6 datasets spanning 5 clinical domain shifts: sequence, disease stage, clinical site, imaging phase, and modality.  \underline{Sequence}: we train on T1, T2, and T1-CE MRIs from MSD-BraTS~\cite{antonelli2022medical} and evaluate on FLAIR from unseen subjects. \underline{Disease-stage}: we train on mild and moderate cardiac disease scans from HVSMR~\cite{pace2024hvsmr} and evaluate on severe cases.  \underline{Site}: we use a multi-site prostate MRI collection~\cite{Bloch2015, lemaitre2015computer, litjens2014evaluation, liu2021feddg}, where splits come from distinct imaging sites. \underline{Phase}: we use PanDG~\cite{zhang2025rethink} to train and evaluate on venous- and out-of-phase abdominal MRIs, respectively, all from distinct sites. \underline{Modality}: we evaluate generalization from CT to MRI on AMOS~\cite{ji2022amos} and from CTA to MRA on TopCoW~\cite{topcowchallenge} for abdominal organs and cerebral vessels, respectively. All other dataset-specific details are in Appendix~\ref{app:dataset_preprocessing}.

\begin{table}[!t]
\caption{
\textbf{Few-shot domain generalization results.}
\textbf{Left}: Mean Dice score (w/ std. error) for few-shot experiments. 
\textbf{Right}: Arbitrarily selected few-shot results.%
}
\centering
\begin{minipage}[t]{0.52\textwidth}
\vspace{0pt}
\centering
\footnotesize
\raggedright
\setlength{\tabcolsep}{3pt} %
\begin{tabular}{l c c c c}
\toprule
\textbf{Method} &
\begin{tabular}{c} \textbf{AMOS}\\ {\scriptsize Modality} \end{tabular} &
\begin{tabular}{c} \textbf{CoW}\\ {\scriptsize Modality} \end{tabular} &
\begin{tabular}{c} \textbf{HVSMR}\\ {\scriptsize Disease} \end{tabular} &
\begin{tabular}{c} \textbf{Prostate}\\ {\scriptsize Site} \end{tabular} \\
\midrule
ERM & 5.4(1.0) & 18.2(1.3) & 40.4(3.8) & 27.8(5.0) \\
GIN & 6.0(0.7) & 13.4(1.1) & \underline{$44.9 (4.0)$} & 20.3(4.3) \\
GINIPA  & $\underline{37.8(1.5)}$ & $23.8$(1.7) & $44.5$(3.8) & $27.9$(5.1) \\
MixUp    & $15.6$(1.6) & $\underline{33.7(1.2)}$ & $43.7$(3.9) & $\underline{53.2(3.4)}$ \\
CutOut & 8.4(1.1) & 18.2 (1.4) & 41.3(3.6) & 40.8(4.5) \\
DSBN     & $19.3$(1.8) & $22.2$(1.5) & $44.2$(3.7) & $44.1$(5.2) \\
MixStyle & 3.8(0.8) & 17.6(1.1) & 37.7(3.5) & 25.0(5.1)\\
RSC & 3.6(1.0) & 18.3(1.3) & 44.1(3.9) & 28.5(5.2) \\
\midrule
Ours     & $\textbf{53.4 (1.4)}$ & $\textbf{40.3 (1.5)}$ & $\textbf{46.3 (3.8)}$ & $\textbf{64.2 (4.4)}$ \\
\bottomrule
\end{tabular}
\end{minipage}%
\hfill%
\begin{minipage}[t]{0.46\linewidth}
\vspace{-4pt}
\raggedright
\includegraphics[width=\linewidth]{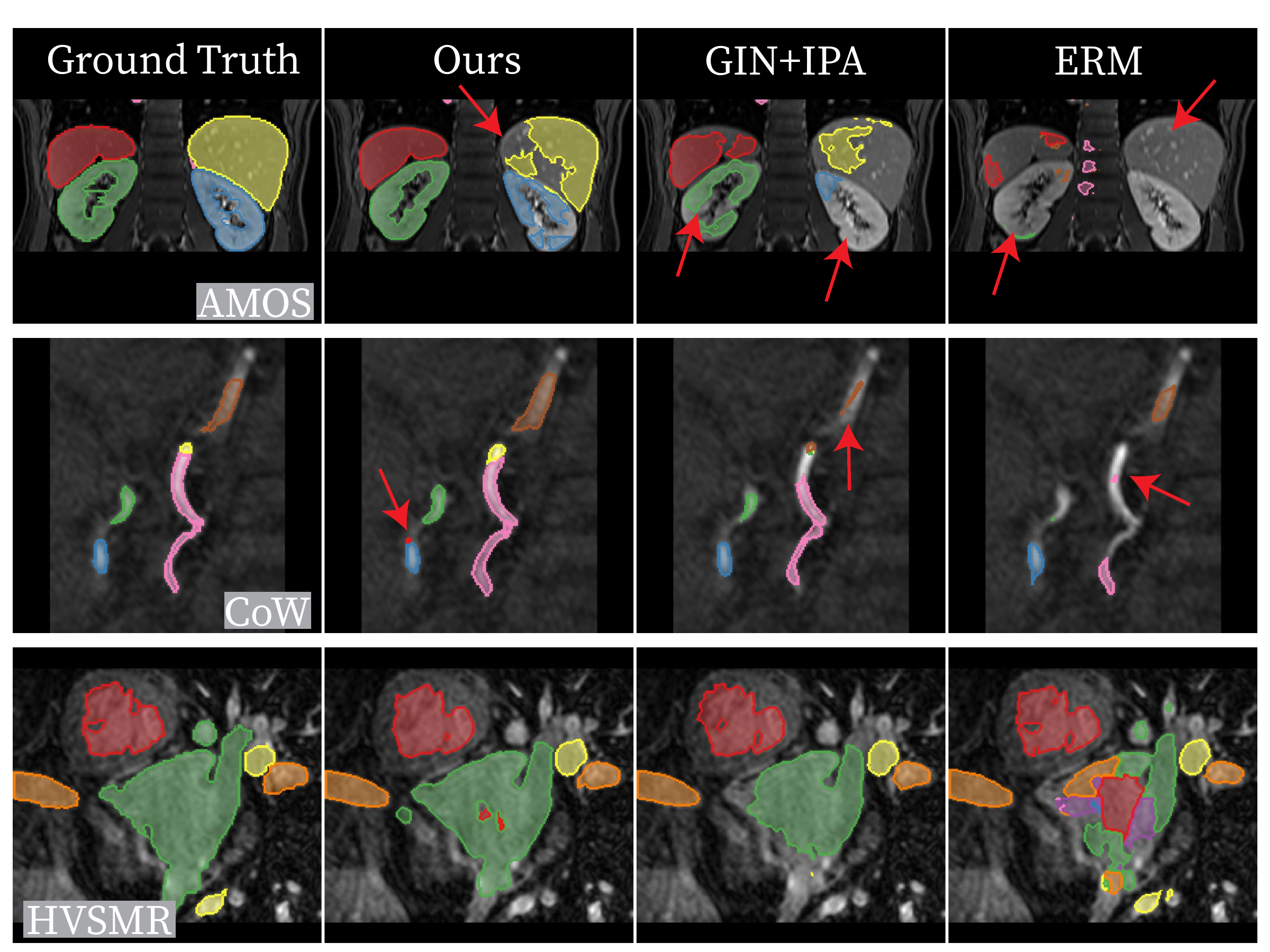}
\end{minipage}
\label{tab:main_results_fewshot}
\end{table}

\subpara{Baselines.}
To reflect practical clinical deployment, we compare against state-of-the-art DG methods that neither require domain labels nor adapt on unlabeled test data. Unless otherwise noted, all methods use the same UNet architecture for their segmentation branch. For fair comparison, we optimize each baseline's primary hyperparameter using a grid search, detailed in Appendix~\ref{app:additional_results}. We apply a strong and extensive augmentation pipeline across all methods to maximize their out-of-distribution generalization potential. 
Our ERM baseline reflects standard practice and only trains on the unstable input using an architecture and training routine matched to ours. We also evaluate multiple data-level (GIN, GIN+IPA~\cite{ouyang2022causality}; MixStyle~\cite{zhou2021domain}; MixUp~\cite{zhang2017mixup}; CutOut~\cite{zhou2021domain}) and model-level (DSBN~\cite{zhou2022generalizable}; RSC\cite{huang2020self}) domain generalization methods, detailed in Appendix~\ref{app:baseline_implementation}. We also compare against the widely used nnUNet~\cite{isensee2021nnu}, using its default architecture and augmentation pipeline.

\subpara{Data Availability Settings.} We evaluate performance in two common medical imaging scenarios: datasets with abundant annotations ("\textbf{all-data}") and those with limited annotations ("\textbf{few-shot}"). The all-data models are trained for 250K iterations following~\cite{isensee2021nnu}, while few-shot models are trained for 40K steps as in~\cite{dey2024learning}. The few-shot experiments use five randomly selected subjects from each dataset for training and exclude the BraTS and PanDG datasets due to high inter-sample heterogeneity (\eg, in tumor presentations) that precludes meaningful training and interpretation of few-shot experiments. %

\begin{figure}[t!]
    \centering
    \includegraphics[width=\linewidth]{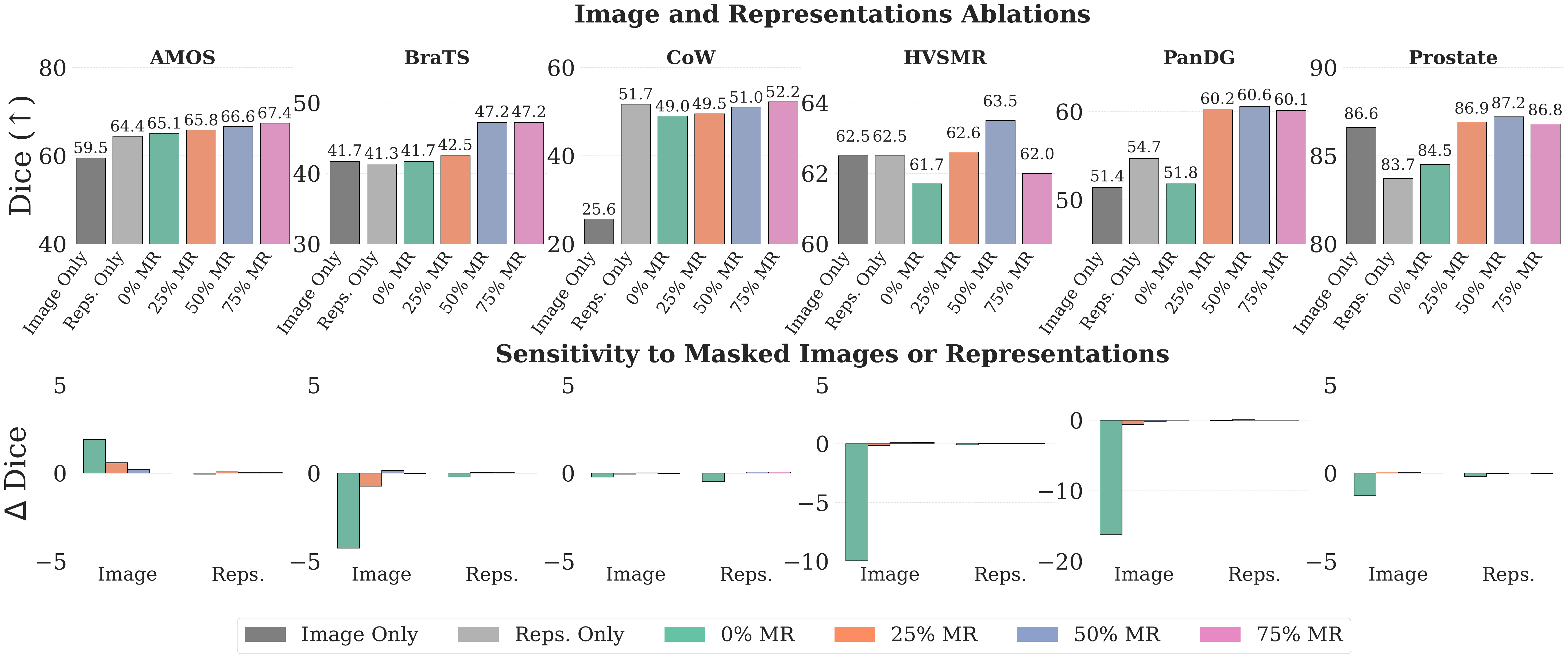}
    \caption{\textbf{Ablating Regularization via View Masking}. 
    \textbf{Top:} We train models with only the image (``Image Only"), only the representations (``Reps. Only"), and both combined with increasing view masking rates (MR=$0\%,25\%$,$50\%$,$75\%$). \methodname  increases validation performance over training on the image alone, the representations alone, and over combination without masking (MR=$0\%$). 
    \textbf{Bottom:} We perform a one-channel removal analysis to measure the sensitivity in validation ($\mathrm{\Delta Dice}$) w.r.t. the full input (image and representations). The drop in performance indicates that, without view masking, the model does not balance information sources and heavily relies on the image input. %
    }
    \label{fig:dropout_sensitivity}
    \vspace{-1em}
\end{figure}

\subsection{Results} \label{subsec:main_results}

\looseness=-1
\subpara{All-data Experiments.} Table~\ref{tab:main_results_full} summarizes the results across all six datasets. Given the diversity of evaluated domain shifts, we also report the average ranking (far right, 1 = best). \methodname achieves an average ranking of 1.5, demonstrating the reliability and consistency of our approach. \methodname performs particularly well under strong appearance differences between source and target domains, as in AMOS, BraTS, and CoW. While the overall runner-up GIN+IPA (average rank of 2.5) performed better on PanDG, our method was second on that dataset and substantially outperformed the remaining methods. For the disease stage shift in HVSMR, most methods converge to within a single Dice point of one another. Notably, our augmentation pipeline lifts performance across all methods as our ERM baseline outperforms the popular nnUNet (which is recommended to be used as is) on four of six datasets, underscoring the outsized role augmentation plays in domain generalization. %

\subpara{Few-shot Experiments.} Table~\ref{tab:main_results_fewshot} and its figure report few-shot segmentation results. \methodname achieved the best performance across all datasets, observing particularly large performance gains over the second-best methods of 15.6, 6.6, 11.0 Dice points on the AMOS, CoW, and Prostate datasets, respectively. Notably, the next-best method changes between datasets, whereas \methodname maintains consistently high performance across domain shift types and anatomical regions, indicating that using regularized combinations of images and representations yields significant data efficiency. Further, while GIN+IPA ranked second in the all-data experiments, it underperformed in the few-shot setting.

\subsection{Analysis \& Ablations}
\label{subsec:analysis_and_ablations}
\subpara{Training on Only Stable Representations.} We train networks on the stable inputs alone without the original input to investigate the benefits of incorporating both stable and unstable features. \textbf{Fig~\ref{fig:dropout_sensitivity}} (Top) shows that training on stable representations alone (``Reps. Only'') has inconsistent performance across datasets, indicating that there is a need to also utilize unstable, in-domain input features.

\looseness=-1
\subpara{Ablating Regularization via View Masking.} We now investigate the regularization effects of masking the combined features. 
As shown in \textbf{Fig.~\ref{fig:dropout_sensitivity}} (Top), concatenating unstable and stable inputs without masking (0\% Mask Rate (MR)) improves validation performance vs. image-only on four of six datasets but causes stagnation or decline in others.
In contrast, applying moderate masking (25--75\% MR) on the combination leads to improved performance, agreeing with our theoretical analyses. Further, \textbf{Fig.~\ref{fig:dropout_sensitivity}} (Bottom) assesses dependence on the original image in \methodname models on validation data by masking either the image or the representation channels during inference. Models trained without masking (0\% MR) exhibit a strong dependence on the image channel, as masking it causes large performance drops on several datasets, whereas masking representations has a relatively smaller impact. With regularization via view masking, dependence on the original image is alleviated, indicating that the model learns to use both image and representation features more effectively.

\begin{wraptable}{r}{0.5\textwidth} %
\vspace{-12pt}
\caption{\textbf{Varying the representation extractor.} 
Validation Dice scores for mask rates of 0.0, 0.25, 0.50, and 0.75, as well as training on  reps. alone.}
\label{tab:cd_ablation_all}
\centering
\scriptsize
\setlength{\tabcolsep}{3pt}
\begin{tabular}{@{}llccccc@{}}
\toprule
Dataset & Model & 0.0 & 0.25 & 0.50 & 0.75 & RepOnly \\
\midrule
\multirow{2}{*}{AMOS\cite{ji2022amos}} 
  & nnInteract.~\cite{isensee2025nninteractive} & 64.3 & 65.0 & 65.9 & 66.4 & 63.9 \\
  & anatomix~\cite{dey2024learning} & 65.1 & 65.8 & 66.6 & \textbf{67.4} & 64.4 \\
\addlinespace[0.5em]
\multirow{2}{*}{BraTS\cite{antonelli2022medical}} 
  & nnInteract.~\cite{isensee2025nninteractive} & 41.5 & 43.4 & 44.2 & 42.9 & 42.7 \\
  & anatomix~\cite{dey2024learning}          & 41.7 & 42.5 & 47.2 & \textbf{47.2} & 41.3 \\
\addlinespace[0.5em]
\multirow{2}{*}{CoW\cite{topcowchallenge}} 
  & nnInteract.~\cite{isensee2025nninteractive} & 48.7 & 45.9 & 48.9 & 50.8 & 50.7 \\
  & anatomix~\cite{dey2024learning}          & 49.0 & 49.5 & 51.0 & \textbf{52.2} & 51.7 \\
\addlinespace[0.5em]
\multirow{2}{*}{HVSMR\cite{pace2024hvsmr}} 
  & nnInteract.~\cite{isensee2025nninteractive} & 61.1 & 61.5 & 63.1 & \textbf{63.9} & 61.8 \\
  & anatomix~\cite{dey2024learning}          & 61.7 & 62.6 & 63.5 & 62.0 & 62.5 \\
\addlinespace[0.5em]
\multirow{2}{*}{PanDG\cite{zhang2025rethink}} 
  & nnInteract.~\cite{isensee2025nninteractive} & 49.0 & 49.7 & 54.6 & 59.0 & 51.2 \\
  & anatomix~\cite{dey2024learning}          & 51.8 & 60.2 & \textbf{60.6} & 60.1 & 54.7 \\
\addlinespace[0.5em]
\multirow{2}{*}{Prostate\cite{liu2021feddg}} 
  & nnInteract.~\cite{isensee2025nninteractive} & 87.6 & 87.7 & \textbf{88.1} & 87.9 & 87.8 \\
  & anatomix~\cite{dey2024learning}          & 84.5 & 86.9 & 87.2 & 86.8 & 83.7 \\
\bottomrule
\end{tabular}
\vspace{-10pt}
\end{wraptable}

\subpara{Changing the Representation Extractor.} We use anatomix~\cite{dey2024learning} as $f_{\phi}$ because of its domain-invariance. However, \methodname is generic and can use other 3D biomedical foundation models as well. We now use nnInteractive~\cite{isensee2025nninteractive} with \methodname as it is the only other pan-application 3D biomedical foundation model.
We extract representations from its penultimate layer and assume no prompts as inputs. %
In \textbf{Table~\ref{tab:cd_ablation_all}}, anatomix achieved higher validation performance in four of six datasets and hence was chosen as  $f_{\phi}$ for the primary experiments. We also find that previous anatomix trends continue to hold for nnInteractive: simple concatenation of images and representations (MR=0\%) does not help, but regularization with MR$>$0\% substantially improves validation performance.

\begin{table}[!t]
\caption{\textbf{Comparisons with methods that perform adaptation.} Mean Dice score for test-time adaptation (TTA) and joint-predictor settings. TTA methods observe the test distribution at inference, whereas \methodname uses no adaptation and shares the same model weights across all evaluations. 
}
\centering
\resizebox{\columnwidth}{!}{%
\setlength{\tabcolsep}{2.5pt}
\begin{tabular}{l c c c c c c c c}
\toprule
\textbf{Method} &
\begin{tabular}{c} \textbf{Adaptation} \\ \footnotesize \textbf{Needed?} \end{tabular} &
\begin{tabular}{c} \textbf{AMOS}\\ \footnotesize Modality \end{tabular} &
\begin{tabular}{c} \textbf{BraTS}\\ \footnotesize Sequence \end{tabular} &
\begin{tabular}{c} \textbf{CoW}\\ \footnotesize Modality \end{tabular} &
\begin{tabular}{c} \textbf{HVSMR}\\ \footnotesize Disease \end{tabular} &
\begin{tabular}{c} \textbf{PanDG}\\ \footnotesize Phase \end{tabular} &
\begin{tabular}{c} \textbf{Prostate}\\ \footnotesize Site \end{tabular} &
\textbf{\begin{tabular}{c} Avg.\\ Rank \end{tabular}} \\
\midrule
EntMin (ERM)\cite{wang2020tent} & Yes & 64.7 (1.2) & 38.9 (1.8) & 24.4 (1.5) & 64.4 (2.9) & 27.7 (2.4) & 84.5  (0.5) & 4.2 \\
EntMin (MR=0\%)\cite{wang2020tent} & Yes & $\underline{65.8 (1.2)}$ & $\underline{39.3 (1.6)}$ & 45.1 (1.2) & 65.0 (2.8) & $\underline{35.7 (2.0)}$ & 83.6 (0.5) & $\underline{3.0}$ \\
SFB\cite{eastwood2023spuriosity} & Yes & 63.6 (1.0) & 38.8 (1.6) & 41.0 (1.3) & 65.2 (2.6) & 30.4 (2.5) & $\textbf{85.4 (0.5)}$ & 3.3 \\
Alpha Tuning~\cite{anti-causal} & Yes & 64.1 (1.3)& 41.7 (1.8) & $\textbf{49.0} (1.4)$ & 64.8 (3.3) & 28.5 (2.5) & 84.3 (0.5) & 3.2 \\
\midrule
Ours  & No & $\textbf{67.5 (0.9)}$  & $\textbf{46.2 (1.6)}$ & $\underline{47.8(1.3)}$ & $\textbf{65.8(3.4)}$ & $\textbf{46.0 (1.9)}$ & $\underline{84.7(0.5)}$ & $\textbf{1.3}$ \\
\bottomrule
\end{tabular}}
\label{tab:tta_table}
\centering
\caption{\textbf{Source Masking vs. Generic Masking.} Comparison of masking to regularize stable and unstable source combination (Ours) vs. generic masking (channel dropout) on all predictor layers.}%
\footnotesize
\begin{tabular}{lccccccc}
\toprule
\textbf{Method} & \textbf{AMOS}\cite{ji2022amos} & \textbf{BraTS}\cite{menze2014multimodal, antonelli2022medical} & \textbf{CoW}\cite{topcowchallenge} & \textbf{HVSMR}\cite{pace2024hvsmr} & \textbf{PanDG}\cite{zhang2025rethink} & \textbf{Prostate}\cite{liu2021feddg}\\
\midrule
Channel Dropout & 67.0 (1.0) & 39.8 (1.5) & 33.0 (1.1) & 65.3 (3.4) & 35.3 (2.5) & 83.6 (0.6) \\
Ours & $\textbf{67.5 (0.9)}$ & $\textbf{46.2 (1.6)}$ & $\textbf{47.8 (1.3)}$ & $\textbf{65.8 (3.4)}$ & $\textbf{46.0 (1.9)}$ & $\textbf{84.7 (0.5)}$\\
\bottomrule
\end{tabular}
\label{tab:layer_dropout}
\vspace{-1em}
\end{table}

\subpara{Test-Time Adaptation (TTA).} 
Our paper considers domain generalization where no adaptation or retraining is allowed on the test domain. To stress-test our method, we also compare it against methods that \textit{do} adapt to unlabeled test data. These include entropy minimization~\cite{wang2020tent} applied to both ERM and the naïve model that uses stable and unstable representations without regularization (MR=0\%). We also evaluate two recent TTA methods that use predictors trained separately on stable and unstable features: Stable Feature Boosting (SFB)~\cite{eastwood2023spuriosity} and Alpha Tuning~\cite{anti-causal}. Notably, all of these methods need training or calibration on the test domain, whereas \methodname \textbf{does not need adaptation}. Despite never observing the test domain, \methodname outperforms both standard TTA and feature combination baselines on 4 out of 6 datasets and achieves the best average rank in \textbf{Table~\ref{tab:tta_table}}. %

\looseness=-1
\subpara{Targeted Input Masking vs. Generic Masking.} 
\methodname uses view masking to regularize the combination of stable and unstable sources, which is a fundamentally different use of masking from generic masking-based regularization. To disentangle the benefits of masking for source combination vs. generic masking, we now consider the case where channel masking applied to all intermediate convolutional layers of $h_{\theta}$ in a standard ERM setting with only image inputs, which reduces to channel dropout~\cite{tompson2015efficient}. Across all datasets in \textbf{Table~\ref{tab:layer_dropout}}, our proposed targeted source-combination regularization via view masking yields better out-of-domain performance than generic channel dropout.

\begin{wraptable}{r}{0.55\textwidth} %
\vspace{-14pt}
\centering
\scriptsize
\caption{\textbf{Effect of Mask Rate.} Validation Dice score as a function of mask rate used in training. Left: all-data results; right: few-shot results.}
\vspace{-0.5em}
\label{tab:dropout_ablation}
\setlength{\tabcolsep}{3pt} %
\begin{tabular}{l c c c c c c c c}
\toprule
\textbf{Dataset} 
& \multicolumn{4}{c}{\textbf{All-data Setting}} 
& \multicolumn{4}{c}{\textbf{Few-shot Setting}} \\
\cmidrule(lr){2-5} \cmidrule(lr){6-9}
 & $0.0$ & $0.25$ & $0.50$ & $0.75$ 
 & $0.0$ & $0.25$ & $0.50$ & $0.75$ \\
\midrule
AMOS\cite{ji2022amos} & 65.1 & 65.8 & 66.6 & \textbf{67.4} & 23.5 & 41.9 & 50.8 & \textbf{53.6} \\
BraTS\cite{antonelli2022medical} & 41.7 & 42.5 & 47.2 & \textbf{47.2} & -- & -- & -- & -- \\
CoW\cite{topcowchallenge} & 49.0 & 49.5 & 51.0 & \textbf{52.2} & 39.0 & \textbf{43.5} & 37.7 & 43.0 \\
HVSMR\cite{pace2024hvsmr} & 61.7 & 62.6 & \textbf{63.5} & 62.0 & 43.4 & 44.4 & \textbf{45.3} & 44.5 \\
PanDG\cite{zhang2025rethink} & 51.8 & 60.2 & \textbf{60.6} & 60.1 & -- & -- & -- & -- \\
Prostate\cite{liu2021feddg} & 84.5 & 86.9 & \textbf{87.2} & 86.8 & 68.4 & 76.8 & \textbf{82.5} & 82.1 \\
\bottomrule
\end{tabular}
\vspace{-13pt}
\end{wraptable}

\looseness=-1
\subpara{Varying Mask Rate.} %
We now evaluate the effect of the mask rate $r \in [0.0, 0.25, 0.50, 0.75]$ across all datasets. \textbf{Table~\ref{tab:dropout_ablation}} reports the average validation Dice scores for each setting. Holistically, a higher $r$ tends to improve validation performance, highlighting the benefits of weight sharing for out-of-distribution generalization. We perform model selection ($r$) on the validation set using this analysis to compute our test set results in Sec.~\ref{subsec:main_results}. However, we note that not tuning \methodname and simply using $r$=0.5 across all datasets still outperforms baselines that \textit{are} tuned on validation data (Appendix~\ref{app:r0.5}), highlighting its plug-and-play nature and ease-of-use across biomedical and clinical settings.

\subpara{Additional Analyses.} Further empirical verification of Condition~\ref{cond:nai}, \methodname's robustness to simulated domain shifts, finetuning $f_{\phi}$ directly, and $f_{\phi}$'s  domain stability is presented in Appendix~\ref{app:additional_results}.

\section{Discussion}
\label{sec:discussion}

\looseness=-1
\subpara{Limitations.} \methodname adds a forward pass through a frozen foundation model as computational overhead. For example, on BraTS, this results in an average training time per iteration increase of 26\% relative to ERM. However, we note that existing domain generalization methods typically incur much higher computational costs. For example, the runner-up in our experiments (GIN+IPA~\cite{ouyang2022causality}) incurs a much larger 450\% training time increase. Further, \methodname's overhead is marginal at inference, as the extra forward pass adds only 8\% runtime relative to ERM. Runtimes are analyzed in Appendix~\ref{app:runtime}. %

\looseness=-1
\subpara{Conclusions.}
This paper presented \methodname, a simple and general framework for domain-generalized biomedical image segmentation that combines raw image intensities with stable representations from foundation models. By enforcing robust usage of both information sources through a lightweight feature combination regularizer, \methodname achieved strong and consistent gains across diverse domain shifts, modalities, anatomies, and data-availability settings, all without requiring complex architectural changes or specialized loss functions used by existing methods. \methodname's simplicity follows from, and is supported by, our theoretical analyses showing that training on partial views of stable and unstable sources yields a principled solution. Our experiments demonstrated that \methodname outperformed state-of-the-art domain generalization methods in both standard training and few-shot regimes, enabling practitioners to no longer need complex data- or model-level interventions to train and use robust segmentation models in their applications. %

\section{Acknowledgments}
Sebo Diaz is supported by the National Science Foundation Graduate Research Fellowship Program (NSF GRFP) under Grant No. DGE-2146755, the MathWorks MATLAB Graduate Fellowship, and the National Institutes of Health (NIH) under Grant No. GENFD0002152100. Polina Golland is partially supported by MIT CSAIL-Wistron Program and MIT Health and Life Sciences
Collaborative (HEALS). Elfar Adalsteinsson is partially supported by NIH R01 EB032708 and NIH R01 EB036945.
Neel Dey is partially supported by NIH NIBIB R01EB033773, NIMH UM1 MH134812-01, and an MGB Radiology Innovation Award.

\bibliographystyle{plain}
\bibliography{main}

\newpage
\appendix
\section{Proofs}
\label{app:proofs}

We first restate the propositions for completeness and then proceed with the proofs.

\subpara{Proposition 1 (Stationarity forces use of stable inputs).}
\textit{Given Assumptions~\ref{ass:stable}--\ref{ass:informative}, let $h_{\theta}$ be a model whose first layer computes: $a^{(1)} = \sigma( W_{u} \star X_{u} + W_{s} \star X_{s} )$
where $W_u, W_s$ denote the first-layer kernel slices corresponding to the unstable and stable input channels, respectively, and are differentiable w.r.t. $\theta$. Suppose the function class $\{h_{\theta}(\mathbf{0}, \cdot) : \theta \in \Theta \}$ can realize the Bayes-optimal predictor for $(X_s, Y)$ and that $\mathcal{R}_{(0,1)}$ has no spurious stationary points w.r.t. $W_s$. Given Condition~\ref{cond:nai}, if $\pi_{(0,1)} > 0$, then at no stationary point $\theta^*$ of $\mathcal{L}$ can $h_{\theta^*}$ ignore $X_{s}$.%
}

\begin{proof}
As a reminder, our input is $X = (X_u, X_s)$ and masks are $(0,1), (1,0), $ or $(1,1)$. Additionally, for clarity, we say $h_{\theta}$ \emph{ignores} $X_s$ if its output under the mask $(0,1)$ is invariant w.r.t. $X_s$, \ie, if $h_{\theta}$ ignores $X_s$, its output is constant w.r.t. changes in $X_s$.

We proceed to construct a proof by contradiction. Let $\theta^{*}$ be a stationary point of $\mathcal{L}$ and let $h_{\theta^*}$ ignore $X_{s}$.

\subpara{Step 1: Input partitioning.} Because the input to $h_{\theta}$ has explicit partitions $(X_u, X_s)$, the first-layer admits a weight slicing into $W_{u}$ and $W_{s}$ by input channel group, while all subsequent layers share weights across the activations produced by both groups, \ie, the input is the only part of the model that is explicitly separated. Stationarity of $\mathcal{L}$ at $\theta^{*}$ requires that the gradient vanishes when restricted to $W_{s}$:
\begin{equation}
\begin{aligned}
\nabla_{W_s}\mathcal{L}(\theta^*)
&= \quad \pi_{1,0}\,\nabla_{W_s}\mathcal{R}_{(1,0)}(\theta^*) \\
&\quad + \pi_{1,1}\,\nabla_{W_s}\mathcal{R}_{(1,1)}(\theta^*) \\
&\quad + \pi_{0,1}\,\nabla_{W_s}\mathcal{R}_{(0,1)}(\theta^*) \\
&= \mathbf{0}.
\end{aligned}
\label{eq:stationary-Ws}
\end{equation}

\subpara{Step 2: The unstable input regime is independent of $W_{s}$.} Under mask $(1,0)$ the stable channels are set to zero, so the first layer computation becomes $\sigma(W_{u} \star X_{u})$, which is independent of $W_{s}$. Because all subsequent layers are functions of this activation alone, the entire forward pass, and hence the loss, is independent of $W_{s}$: $\nabla_{W_{s}}\mathcal{R}_{(1,0)}(\theta^*)=0.$
This holds exactly by the bi-linearity of cross-correlation: $W_{s} \star \mathbf{0} = \mathbf{0}$ and $\partial(W_{s} \star \mathbf{0}) = \mathbf{0}.$ Substituting this into Eq.~\ref{eq:stationary-Ws} eliminates the first term.

\subpara{Step 3: Stable-only risk is suboptimal.} Under mask $(0,1)$ the input is $(\mathbf{0},X_s)$. Because $h_{\theta^*}$ ignores $X_s$, the output is constant in $X_s$. For any strictly proper scoring rule, the minimal achievable risk by a constant predictor is $H(Y)$, so $\mathcal{R}_{(0,1)}(\theta^*) \geq H(Y)$.

By Assumption~\ref{ass:informative}, $I(Y; X_s) > 0$, hence $H(Y \mid X_s) < H(Y)$. Therefore $\mathcal{R}_{(0,1)}(\theta^*) \geq H(Y) > H(Y \mid X_s)$, which violates the hypothesis that $\mathcal{R}_{(0,1)}$ has no suboptimal stationary points. It follows that $\theta^*$ is not stationary for $\mathcal{R}_{(0,1)}$ w.r.t.\ $W_s$, and since the risk is differentiable, there exists a direction $v_s$ such that
\[
\langle \nabla_{W_s} \mathcal{R}_{(0,1)}(\theta^*),\, v_s \rangle < 0.
\]

\subpara{Step 4: Contradiction via non-adversarial interaction.}
Set the gradient direction $v_{s} := -\nabla_{W_{s}} \mathcal{R}_{(0,1)}(\theta^*)$, as a descent direction for the stable-only risk. Because the deeper layers are shared, activating $W_{s}$ also affects predictions under $(1,1)$, so $\nabla_{W_{s}}\mathcal{R}_{(1,1)}(\theta^*)$ does not vanish. However, by Condition~\ref{cond:nai}, the gradient does not oppose the stable-only descent:
\begin{align}
    \bigl\langle \nabla_{W_s}\mathcal{L}(\theta^*),\, v_s
    \bigr\rangle
    &= \pi_{(0,1)}\,
       \underbrace{
         \bigl\langle \nabla_{W_s}\mathcal{R}_{(0,1)}(\theta^*),\,
         v_s \bigr\rangle
       }_{
         \;-\,\bigl\|\nabla_{W_s}\mathcal{R}_{(0,1)}\bigr\|^2
         \;<\; 0
       }
    \;+\;
    \pi_{(1,1)}\,
       \underbrace{
         \bigl\langle \nabla_{W_s}\mathcal{R}_{(1,1)}(\theta^*),\,
         v_s \bigr\rangle
       }_{
         \leq\; 0
         \;\;\text{(Condition~\ref{cond:nai})}
       } \nonumber \\[4pt]
    &< 0,
    \label{eq:contradiction}
\end{align}
which contradicts $\nabla_{W_{s}}\mathcal{L}(\theta^*) = \mathbf{0}$. This concludes the proof.
\end{proof}

\subpara{Proposition 2 (Stable-only performance ceiling).}
\textit{Under Assumptions~\ref{ass:stable}--\ref{ass:informative}, $H(Y | X_s)$ is environment-invariant, and the Bayes-optimal stable-only predictor $h_s^*(x_s) := \mathbb{E}[Y | X_s = x_s]$ achieves this value as its risk in every environment. Therefore, in any environment $e$, the Bayes-optimal joint risk $H_e(Y | X_s, X_u) < H(Y | X_s)$, with strict inequality, where $H_e(\cdot)$ denotes entropy under $\mathbb{P}_e$.} %

\begin{proof}
We provide the following proof in three steps.

\subpara{Step 1: Environment-invariance of $H(Y \mid X_s)$.}
By Assumption~\ref{ass:stable}, $Y \perp E \mid X_s$, so
for any environment $e$,
$\mathbb{P}_e(Y \mid X_s) = \mathbb{P}(Y \mid X_s)$.
It follows that
$H_e(Y \mid X_s) = H(Y \mid X_s)$ for all
$e \in \mathcal{E}$. The Bayes-optimal stable-only
predictor $h_s^*(x_s) := \mathbb{E}[Y \mid X_s = x_s]$
therefore achieves risk $H(Y \mid X_s)$ in every
environment, including unseen ones.

\subpara{Step 2: Joint risk is not larger than the stable-only risk.}
For any environment $e$, the conditional entropy satisfies
\[
    H_e(Y \mid X_s, X_u)
    = H_e(Y \mid X_s) - I_e(Y;\, X_u \mid X_s)
    \leq H_e(Y \mid X_s)
    = H(Y \mid X_s),
\]
where the inequality follows because mutual information
is non-negative, and the final equality from Step 1.

\subpara{Step 3: Strict inequality between joint and stable-only risk.}
By Assumption~\ref{ass:complementarity}, 
$Y \not\perp_e X_u \mid X_s$ for every $e$, 
so $I_e(Y;\, X_u \mid X_s) > 0$ and hence
\[
    \Delta_e 
    := H(Y \mid X_s) - H_e(Y \mid X_s, X_u) 
    \geq I_e(Y;\, X_u \mid X_s) > 0.
\]
Therefore, $H_e(Y \mid X_s, X_u) < H(Y \mid X_s)$, concluding the proof.
\end{proof}

\clearpage
\newpage

\section{Additional Results}
\label{app:additional_results}

\begin{figure}[!th]  %
    \centering
    \includegraphics[width=0.95\textwidth]{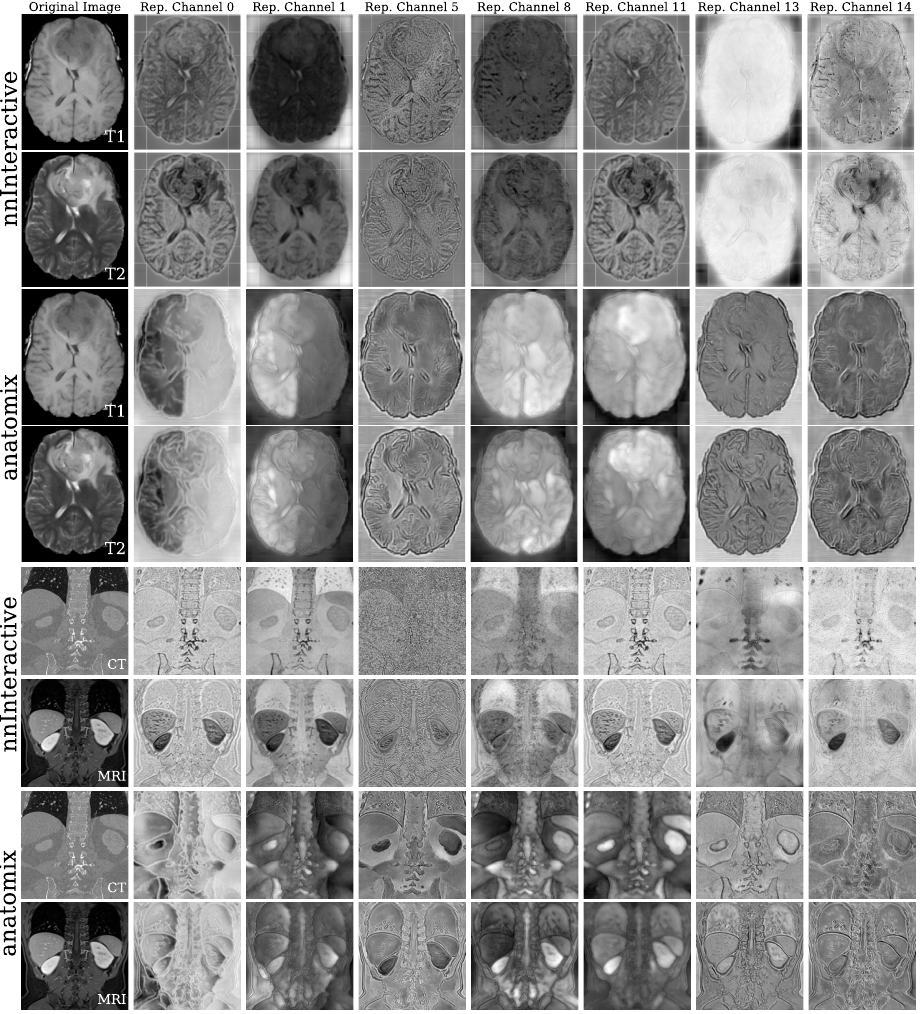}
\caption{
\textbf{Comparing cross-modality representations extracted by foundation models.} \textit{Rows 1--4:} Given a subject from BraTS, we visualize its representations across MRI sequence shifts produced by nnInteractive~\cite{isensee2025nninteractive} (rows 1, 2) and anatomix~\cite{dey2024learning} (rows 3, 4), finding that anatomix representations are more stable and suitable for domain generalization. \textit{Rows 5--8:} We now produce a similar visualization for two unpaired subjects with cross-modality domain shift from the AMOS dataset and find the same stability trend.
}
    \label{fig:amos_representations}
\end{figure}

\subsection{Representation Analysis.}
\label{app:representation_analysis}
While the main text visualizes representation channels as RGB images, here we present a broader subset of arbitrarily selected representations to provide a more comprehensive view. 

\subpara{Representation Stability and Choice of Feature Extractor.}
In addition to the quantitative results shown in Table 4 of the main text, we now visualize the representations extracted by the two foundation models we considered (anatomix~\cite{dey2024learning} and nnInteractive~\cite{isensee2025nninteractive}) to qualitatively investigate their stability to domain shifts. In Figure~\ref{fig:amos_representations}, we compare T1 vs. T2 and CT vs. MRI features from BraTS and AMOS, respectively, to assess the consistency of learned features across different imaging modalities. In this comparison, anatomix produces visually more stable and coherent representations under domain shifts, even though nnInteractive has been exposed to a substantially larger corpus of real-world data. This suggests that anatomix may better capture domain-stable features, as needed in our work. 

\begin{figure}[!tb]
  \centering
  \includegraphics[width=1.0\columnwidth]{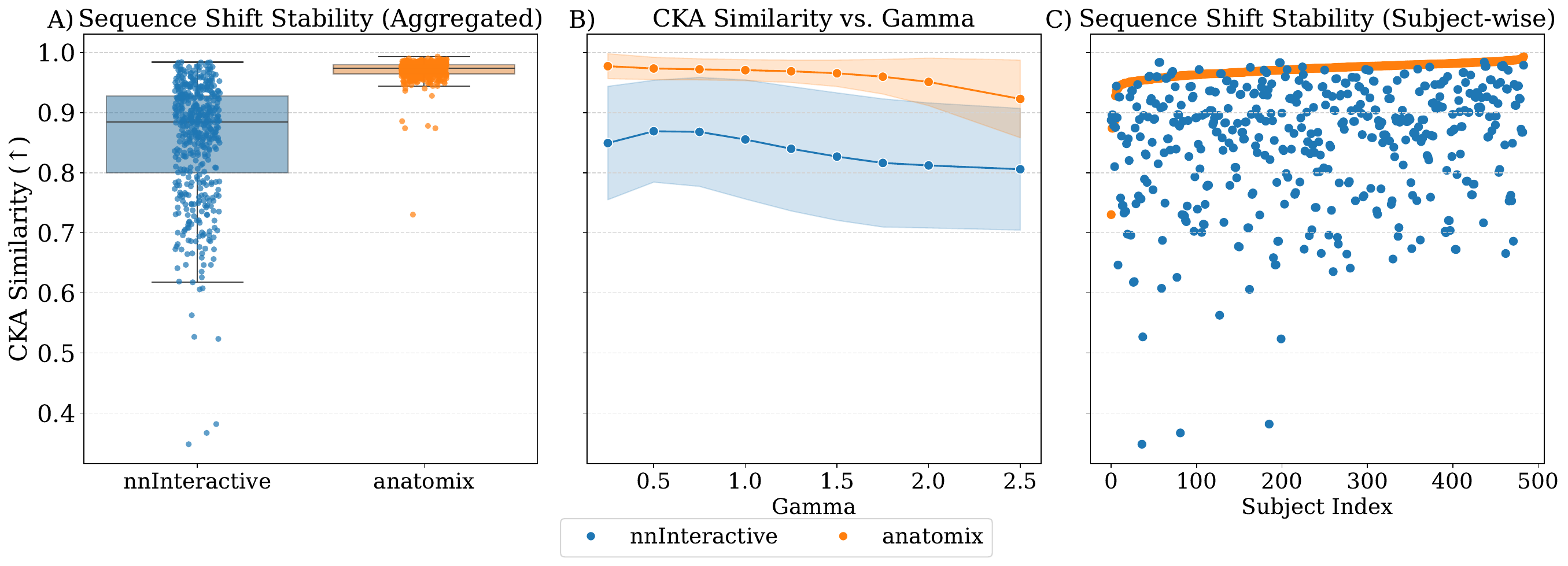}
    \caption{\textbf{CKA Representation Similarity of Foundation Models Under Sequence Shifts.} \textit{A:} CKA~\cite{kornblith2019similarity} between each 3D foundation model's chosen embeddings across MRI sequence shifts given subjects with multiple MRI sequences from BraTS. \textit{B:} Given a fixed sequence, we now perturb image contrast (via gamma adjustments) and visualize CKA as a measure of representational stability under contrast shifts. \textit{C:} A per-subject visualization of the plot in panel \textit{A}.}
  \label{fig:cka_similarity}
\end{figure}

To quantitatively measure representation stability, we compute the centered kernel alignment (CKA) similarity~\cite{kornblith2019similarity} (where higher is more similar) between the representations extracted from co-registered FLAIR and T1 subjects in BraTS. We plot the CKA results for all subjects in Figure~\ref{fig:cka_similarity}A and C, and find that anatomix, being specifically trained for domain stability is stable to these shifts. Further, given random contrast perturbation with Gamma adjustment in Figure~\ref{fig:cka_similarity}B, anatomix maintains robustness across corruption levels. These reasons together demonstrate why we choose anatomix~\cite{dey2024learning} as a stable representation extractor for \methodname over other current foundation models for biomedical volumes.

\subpara{Choice of layer.} Representations can be extracted from arbitrary layers of foundation models. While the original anatomix model~\cite{dey2024learning} uses the final network layer as representations, we conduct a preliminary qualitative analysis for alternative representations. We examine qualitative differences between the penultimate and final block outputs for intra-subject FLAIR vs. T2 sequences and CT vs. MRI abdomen (Figure~\ref{fig:brats_representations}), from the BraTS and AMOS datasets, respectively. The illustrations indicate that, across either FLAIR vs. T2 or CT vs. MR, representations remain relatively stable. However, outputs from the penultimate block retain finer structural details, such as edges and high-frequency information, likely due to contributions from skip connections. In contrast, the final block tends to produce more abstracted representations that may lose some of these fine details. Based on these observations, we chose to use the penultimate block outputs in our method, as representations with stronger edges that are cross-modally consistent are likely to be useful for segmentation.

\begin{figure}[!t]
    \centering
    \includegraphics[width=\linewidth]{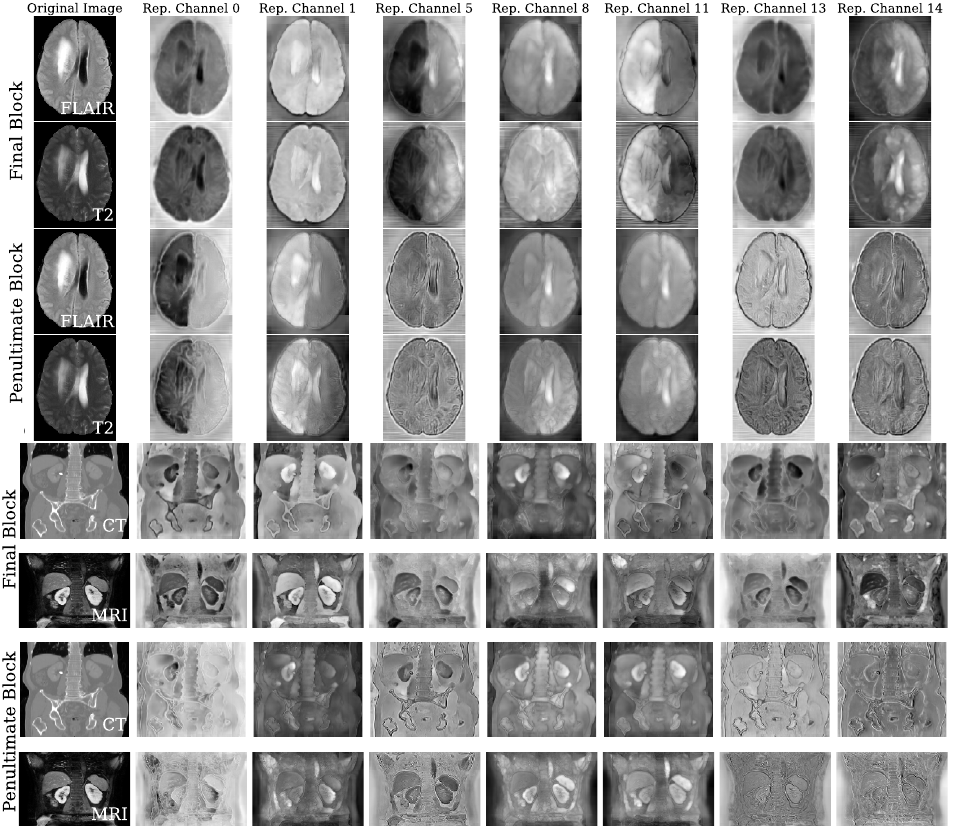}
    \caption{\textbf{Choice of layer from which to use representations.} 
We compare anatomix's final block representations (rows 1–2, 5-6) to its penultimate block representations (rows 3–4, 7-8) for two different datasets with two different kinds of domain shifts. \textit{Rows 1--4:} FLAIR (rows 1 and 3) and T2 (rows 2 and 4) from BraTS, representing intra-subject MRI sequence shifts. \textit{Rows 5--8:} CT (rows 5 and 7) and MRI (rows 6 and 8) from AMOS, representing inter-subject modality shifts. Representation channels (columns 2-7) were arbitrarily selected. Qualitatively, the penultimate output better preserves high-frequency details while retaining inter-domain stability.}
    \label{fig:brats_representations}
\end{figure}

\subsection{Empirical Validation of Condition~\ref{cond:nai}}
\label{app:condition_1}
To empirically support Condition~\ref{cond:nai}, we examine the gradients of the initial convolutional kernel corresponding to the stable inputs. Specifically, we consider the regime where we mask the original input versus the regime where we mask nothing. Condition~\ref{cond:nai} predicts that these gradients have non-conflicting directions, \ie, their inner product is $\ge 0$. We plot a histogram of the cosine similarities (to bound values between -1 and 1) between the joint training gradient $\nabla_{W_{s}}\mathcal{R}_{(1,1)}$ and the stable-only gradient $\nabla_{W_{s}}\mathcal{R}_{(0,1)}$, where $W_{s}$ denotes the kernel slices of the first convolutional layer corresponding to the stable input channels. The cosine similarity is used as a normalized surrogate for the inner product. We compute this quantity every 100 training steps on a held-out batch, evaluating the model in inference mode to avoid corrupting batch-normalization statistics. Figure~\ref{fig:cond_cosine_similarity} shows that the overwhelming majority of measurements fall well above zero throughout training, providing empirical support that the joint training gradient does not oppose the stable-input gradient. 

\begin{figure}[!t]
    \centering
    \includegraphics[width=\linewidth]{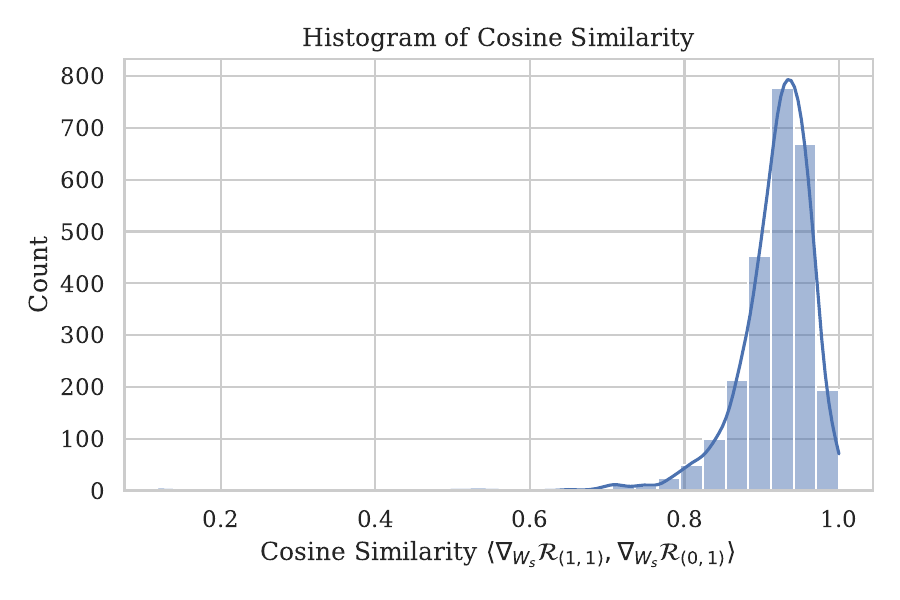}
    \caption{\textbf{Histogram of Cosine Similarities.} Cosine similarities for the two masked regimes of (1,0) and (1,1) for our method sampled every 100 steps during an "All-data" training experiment  on BraTS~\cite{menze2014multimodal, antonelli2022medical} to empirically support Condition~\ref{cond:nai}. A KDE curve is fitted to the histogram for clarity.}

    \label{fig:cond_cosine_similarity}
\end{figure}

\subsection{Finetuning the Foundation Model $f_{\phi}$}
Our proposed approach uses frozen foundation model features for training. However, foundation model weights are good initializations for training as well. We now investigate the standard approach of simply finetuning the foundation model $f_{\phi}$ instead of our proposed approach and examine its robustness to domain shifts. We finetune the anatomix~\cite{dey2024learning} model ($f_{\phi}$ in the main text) in the all-data setting. In particular, we train in two regimes: (1) finetuning only the foundation model's decoder, while maintaining a frozen encoder (F.T. Decoder), and (2) finetuning the entire model (F.T. Extractor). 

As seen in Table~\ref{tab:finetune_extractor}, our method significantly outperforms each of these approaches on every dataset and further motivates our approach. This result is intuitive as finetuning the foundation model using an in-domain ERM loss would not guarantee any robustness on distribution-shifted test samples, as shown by our theoretical analyses.

\begin{table}[t]
\centering
\caption{\textbf{Finetuning domain-stable foundation models does not yield domain generalization}. Ablation comparing finetuning all of the foundation model's weights (F.T. Extractor), finetuning its decoder-alone (F.T. Decoder), and our proposed approach in the All-data setting.}
\scriptsize
\begin{tabular}{lccccccc}
\toprule
\textbf{Method} & \textbf{AMOS}\cite{ji2022amos} & \textbf{BraTS}\cite{menze2014multimodal, antonelli2022medical} & \textbf{CoW}\cite{topcowchallenge} & \textbf{HVSMR}\cite{pace2024hvsmr} & \textbf{PanDG}\cite{zhang2025rethink} & \textbf{Prostate}\cite{liu2021feddg}\\
\midrule
F.T. Extractor & 56.8 (1.5) & 37.7 (1.7) & 17.3 (1.7) & 64.6 (3.2) & 23.7 (2.6) & 84.2 (0.5) \\
F.T. Decoder & 57.2 (1.5) & 39.5 (1.7) & 18.5 (1.7) & 63.2 (2.9) & 27.7 (2.6) & 83.8 (0.7) \\
Ours & $\textbf{67.5 (0.9)}$ & $\textbf{46.2 (1.6)}$ & $\textbf{47.8 (1.3)}$ & $\textbf{65.8 (3.4)}$ & $\textbf{46.0 (1.9)}$ & $\textbf{84.7 (0.5)}$\\
\bottomrule
\end{tabular}
\label{tab:finetune_extractor}
\end{table}

\subsection{Alternative Feature Combination/Fusion Approaches.}
Our proposed method applies feature concatenation and view masking to the stable and unstable features. We now perform additional experimentation on select datasets to investigate alternative feature combination methods in the fully supervised setting and examine their ability to generalize across domains. Specifically, we train a method that uses Squeeze+Excite~\cite{1709.01507} layers (a form of channel attention) to combine extracted representations and the original image before feeding it into the UNet. 
As in Table~\ref{tab:feature_fusion_methods}, our proposed view masking approach outperforms channel attention on our domain generalization benchmarks as channel attention layers trained on in-domain ERM losses cannot ensure generalization to out-of-domain test data.

\begin{table}[h]
\centering
\caption{\textbf{Alternative approaches for feature combination.} A comparison of alternative feature combination methods: squeeze+excite channel attention vs. our input view masking.}%
\footnotesize
\begin{tabular}{lccccccc}
\toprule
\textbf{Ablations} & \textbf{AMOS}\cite{ji2022amos} & \textbf{BraTS}\cite{menze2014multimodal, antonelli2022medical} \\
\midrule
Squeeze+Excite\cite{1709.01507} & 64.3 (1.1) & 38.8 (1.7) \\
Ours & $\textbf{67.5 (0.9)}$ & $\textbf{46.2 (1.6)}$ \\
\bottomrule
\end{tabular}
\label{tab:feature_fusion_methods}
\end{table}

\subsection{Baseline Hyperparameter Tuning.} For a robust comparison, wherever applicable, we tuned each baseline's perturbation frequency on each dataset's validation set for both the fully supervised ("\textbf{all-data}") and ("\textbf{few-shot}") settings. Specifically, each baseline typically applied its proposed method with stochastic probability $p$. For GIN+IPA\cite{ouyang2022causality}, we further swept values of 1, 10, 20 for its KL divergence consistency loss term (the original recommended value was 10). The hyperparameter(s) leading to the highest performance on the validation set was used to report test set results in Tables~\ref{tab:main_results_full} and~\ref{tab:main_results_fewshot}. Validation set tuning results can be found in Tables~\ref{tab:baseline_tuning_full} and~\ref{tab:baseline_tuning_fs}.
\begin{table}[!ht]
\centering
 
\begin{minipage}[t]{0.49\textwidth}
\centering
\caption{\textbf{All-data hyperparameter tuning.} Mean Dice score for each applicable baseline on each dataset. \textbf{Bolded} values are the highest in each dataset. $k_l$ is the KL divergence consistency loss weight and $p$ denotes the probability of applying each method.}
\label{tab:baseline_tuning_full}
\tiny
 
\scriptsize
\setlength{\tabcolsep}{1.8pt}
\begin{tabular}{lccc@{\hskip 0.5em}ccc}
\toprule
 & \multicolumn{3}{c}{AMOS} & \multicolumn{3}{c}{BraTS} \\
\cmidrule(lr){2-4} \cmidrule(lr){5-7}
 & $k_l$=5 & $k_l$=10 & $k_l$=20 & $k_l$=5 & $k_l$=10 & $k_l$=20 \\
\midrule
GIN+IPA   & 65.8 & 65.7 & \textbf{66.3} & 45.9  & 45.5 & \textbf{46.1} \\
\midrule
 & $p$=0.25 & $p$=0.50 & $p$=0.75 & $p$=0.25 & $p$=0.50 & $p$=0.75 \\
\midrule
MixUp     & 60.7  & 60.8 & \textbf{61.8} & 42.3 & 41.1 & \textbf{43.2} \\
CutOut    & \textbf{59.9} & 57.2 & 57.9 & 40.8  & 41.6 & \textbf{42.2} \\
DSBN      & \textbf{61.4} & 59.1 & 58.2 & 42.3 & 41.8 & \textbf{43.1} \\
MixStyle  & \textbf{58.9} & 57.8 & 57.9 & \textbf{41.9} & 41.2 & 41.5 \\
RSC       & 58.7 & \textbf{61.3} & 59.2 & 42.8 & 40.4 & \textbf{43.9} \\
\midrule
\midrule
 & \multicolumn{3}{c}{CoW} & \multicolumn{3}{c}{HVSMR} \\
\cmidrule(lr){2-4} \cmidrule(lr){5-7}
 & $k_l$=5 & $k_l$=10 & $k_l$=20 & $k_l$=5 & $k_l$=10 & $k_l$=20 \\
\midrule
GIN+IPA   & 42.5 & \textbf{43.4} & 43.0 & 61.0 & \textbf{61.2} & 60.8 \\
\midrule
 & $p$=0.25 & $p$=0.50 & $p$=0.75 & $p$=0.25 & $p$=0.50 & $p$=0.75 \\
\midrule
MixUp     & 30.9 & 37.6 & \textbf{41.5} & 63.1 & \textbf{64.4} & 62.7 \\
CutOut    & \textbf{28.6} & 24.3 & 23.4 & \textbf{63.5} & 63.5 & 62.6 \\
DSBN      & 31.9 & \textbf{33.1} & 24.5 & \textbf{64.0} & 63.3 & 63.7 \\
MixStyle  & 20.9 & \textbf{30.9} & 20.4 & 62.4 & 62.1 & \textbf{62.7} \\
RSC       & 18.1 & \textbf{30.8} & 26.5 & 62.6 & 62.3 & \textbf{63.4} \\
\midrule
\midrule
 & \multicolumn{3}{c}{PanDG} & \multicolumn{3}{c}{Prostate} \\
\cmidrule(lr){2-4} \cmidrule(lr){5-7}
 & $k_l$=5 & $k_l$=10 & $k_l$=20 & $k_l$=5 & $k_l$=10 & $k_l$=20 \\
\midrule
GIN+IPA   & \textbf{54.9} & 54.4 & 54.4 & 86.0 & \textbf{86.3} & 85.9 \\
\midrule
 & $p$=0.25 & $p$=0.50 & $p$=0.75 & $p$=0.25 & $p$=0.50 & $p$=0.75 \\
\midrule
MixUp     & 59.9 & 58.4 & \textbf{60.0} & \textbf{86.1} & 85.8 & 83.9 \\
CutOut    & 52.8 & 53.9 & \textbf{57.4} & 86.7 & 86.8 & \textbf{87.3} \\
DSBN      & \textbf{55.9} & 53.8 & 54.6 & 86.9 & 86.3 & \textbf{87.0} \\
MixStyle  & 53.0 & \textbf{53.2} & 52.7 & \textbf{86.7} & 86.6 & 86.6 \\
RSC       & 52.1 & 53.0 & \textbf{54.6} & 86.9 & \textbf{86.9} & 86.6 \\
\bottomrule
\end{tabular}
\end{minipage}
\hfill
\begin{minipage}[t]{0.49\textwidth}
\centering
\captionof{table}{\textbf{Few-shot hyperparameter tuning.} Mean Dice score for each applicable baseline on each dataset. \textbf{Bolded} values are the highest in each dataset. $k_l$ is the KL divergence consistency loss weight and $p$ denotes the probability of applying each method.}
\label{tab:baseline_tuning_fs}
\vspace{6pt}
 
\tiny
 
\scriptsize
\setlength{\tabcolsep}{1.8pt}
\begin{tabular}{lccc@{\hskip 0.5em}ccc}
\toprule
 & \multicolumn{3}{c}{AMOS} & \multicolumn{3}{c}{CoW} \\
\cmidrule(lr){2-4} \cmidrule(lr){5-7}
 & $k_l$=5 & $k_l$=10 & $k_l$=20 & $k_l$=5 & $k_l$=10 & $k_l$=20 \\
\midrule
GIN+IPA   & 5.9 & 39.6 & 25.6 & 21.0 & 17.0 & 28.0 \\
\midrule
 & $p$=0.25 & $p$=0.50 & $p$=0.75 & $p$=0.25 & $p$=0.50 & $p$=0.75 \\
\midrule
MixUp     & 11.6 & 16.4 & \textbf{18.0} & 18.8 & 16.2 & \textbf{35.8} \\
CutOut    & \textbf{10.8} & 6.9 & 3.9 & \textbf{20.4} & 15.0 & 17.8 \\
DSBN      & 18.0 & \textbf{20.0} & 9.1 & \textbf{24.6} & 22.1 & 14.7 \\
MixStyle  & 3.4 & \textbf{5.0} & 3.4 & 15.5 & 16.5 & \textbf{17.9} \\
RSC       & 3.4 & 2.5 & \textbf{4.7} & 12.6 & 13.5 & \textbf{19.4} \\
\midrule
\midrule
 & \multicolumn{3}{c}{HVSMR} & \multicolumn{3}{c}{Prostate} \\
\cmidrule(lr){2-4} \cmidrule(lr){5-7}
 & $k_l$=5 & $k_l$=10 & $k_l$=20 & $k_l$=5 & $k_l$=10 & $k_l$=20 \\
\midrule
GIN+IPA   & 42.5 & \textbf{43.4} & 43.0 & 61.0 & \textbf{61.2} & 60.8 \\
\midrule
 & $p$=0.25 & $p$=0.50 & $p$=0.75 & $p$=0.25 & $p$=0.50 & $p$=0.75 \\
\midrule
MixUp     & 45.4 & \textbf{46.5} & 43.8 & \textbf{79.0} & 78.4 & 78.5 \\
CutOut    & \textbf{41.5} & 39.9 & 40.1 & 69.6 & \textbf{72.3} & 67.0 \\
DSBN      & 42.9 & \textbf{45.0} & 44.5 & \textbf{71.9} & 49.1 & 56.1 \\
MixStyle  & \textbf{41.1} & 40.2 & 41.0 & 65.8 & 60.8 & \textbf{67.3} \\
RSC       & 39.6 & \textbf{43.8} & 42.0 & 60.8 & \textbf{65.3} & 63.3 \\
\bottomrule
\end{tabular}
\end{minipage}
\end{table}

\subsection{\methodname Hyperparameter Sensitivity}
\label{app:r0.5}
\methodname has a single hyperparameter $r$ corresponding to the masking rate. While our primary results on the test splits in Tables~\ref{tab:main_results_full} and~\ref{tab:main_results_fewshot} used choices of $r$ tuned on each dataset's validation split, such tuning is not strictly necessary for \methodname to outperform existing domain generalization baselines. For example, if validation data is unavailable, simply setting $r=0.5$ corresponding to 50\% masking is intuitive, as this provides balanced views of unstable inputs during training. We report test set results when just using $r=0.5$ for both the "all-data" and "few-shot" settings across all datasets as "Ours ($r=0.5$)" in Tables \ref{tab:main_results_full_r0.5} and \ref{tab:main_results_fewshot_r0.5}, respectively. Despite not tuning \methodname's sole hyperparameter, it outperforms existing domain generalization baselines in the majority of settings and maintains the best average ranking. %
We also emphasize that all baselines in Tables~\ref{tab:main_results_full_r0.5} and~\ref{tab:main_results_fewshot_r0.5} have their hyperparameters tuned on validation data, whereas ``Ours ($r=0.5$)" is untuned, constituting an extreme stress test resulting in smaller margins.

\begin{table}[!t]
\caption{\textbf{All-data domain generalization results for \textit{untuned} masking rate.} Mean Dice score (w/ std. error) across datasets against domain generalization baselines, with avg. ranking on the far right. Companion to Table~\ref{tab:main_results_full} with a default untuned masking rate of $r=0.5$.}
\centering
\footnotesize
\setlength{\tabcolsep}{2.5pt}
\begin{tabular}{l c c c c c c c}
\toprule
\textbf{Method} &
\begin{tabular}{c} \textbf{AMOS}\cite{ji2022amos}\\ \footnotesize Modality \end{tabular} &
\begin{tabular}{c} \textbf{BraTS}\cite{antonelli2022medical}\\ \footnotesize Sequence \end{tabular} &
\begin{tabular}{c} \textbf{CoW}\cite{topcowchallenge}\\ \footnotesize Modality \end{tabular} &
\begin{tabular}{c} \textbf{HVSMR}\cite{pace2024hvsmr}\\ \footnotesize Disease \end{tabular} &
\begin{tabular}{c} \textbf{PanDG}\cite{zhang2025rethink}\\ \footnotesize Phase \end{tabular} &
\begin{tabular}{c} \textbf{Prostate}\cite{liu2021feddg}\\ \footnotesize Site \end{tabular} &
\textbf{\begin{tabular}{c} Avg.\\ Rank \end{tabular}} \\
\midrule
ERM & 57.4 (1.7) & 38.9 (1.8) & 24.6 (1.5) & 63.6 (3.0) & 25.3 (2.5) & 83.0 (0.7) & 8.0 \\
GIN\cite{ouyang2022causality} & 65.1 (0.7) & 29.9 (1.6) & 17.9 (1.5) & 64.1 (3.0) & 41.1 (1.8) & 82.1 (0.8) & 7.2 \\
GIN+IPA\cite{ouyang2022causality} & $\underline{66.5 (0.8)}$ & $\underline{44.7 (1.6)}$ & $\underline{37.6 (1.9)}$ & 64.7 (2.6) & $\textbf{50.3 (1.4)}$ & $\underline{84.1 (0.6)}$ & $\underline{2.5}$ \\ 
MixUp\cite{zhang2017mixup} & 61.5 (0.9) & 41.2 (1.7) & 36.3 (1.5) & 64.5 (3.2) & 26.7 (2.5) & 82.5 (0.7) & 5.5 \\
CutOut\cite{devries2017improved} & 55.5 (1.7) & 41.8 (1.6) & 26.5 (1.9) & 64.9 (3.5) & 30.5 (2.4) & 84.0 (0.6) & 5.0 \\
DSBN\cite{zhou2022generalizable} & 58.0 (1.4) & 40.5 (1.7) & 29.6 (1.9) & \underline{$66.1 (3.2)$} & 26.3 (2.6) & 83.4 (0.6) & 5.2 \\
MixStyle\cite{zhou2021domain} & 53.9 (1.7) & 40.7 (1.6) & 27.9 (1.4) & 64.4 (2.6) & 24.0 (2.5) & 82.3 (0.7) & 7.7 \\
RSC\cite{huang2020self}      & 60.8 (1.3) & 42.1 (1.6) & 26.0 (1.9) & $\textbf{66.2 (2.9)}$ & 28.2 (2.4) & 83.4 (0.6) & 4.3 \\
nnUNet\cite{isensee2021nnu}   & 48.0 (2.7) & 35.5 (1.8) & 13.7 (2.0) & 65.0 (2.5) & 30.1 (2.7) & 71.6 (3.6) & 8.0  \\
\midrule
Ours ($r=0.5$) & $\textbf{66.8 (1.0)}$ & $\textbf{44.8 (1.5)}$ & $\textbf{47.2 (1.3)}$ & 65.8 (3.4) & $\underline{46.0  (1.9)}$ & $\textbf{84.7 (0.5)}$ & $\textbf{1.5}$ \\
\bottomrule
\end{tabular}
\label{tab:main_results_full_r0.5}
\end{table}

\begin{table}[!t]
\caption{
\textbf{Few-shot domain generalization results for \textit{untuned} masking rate.} Mean Dice score (w/ std. error) across datasets against domain generalization baselines, with avg. ranking on the far right. Companion to Table~\ref{tab:main_results_fewshot} with a default untuned masking rate of $r=0.5$.
}
\centering
\footnotesize
\setlength{\tabcolsep}{3pt}
\begin{tabular}{l c c c c c}
\toprule
\textbf{Method} &
\begin{tabular}{c} \textbf{AMOS}\\ {\scriptsize Modality} \end{tabular} &
\begin{tabular}{c} \textbf{CoW}\\ {\scriptsize Modality} \end{tabular} &
\begin{tabular}{c} \textbf{HVSMR}\\ {\scriptsize Disease} \end{tabular} &
\begin{tabular}{c} \textbf{Prostate}\\ {\scriptsize Site} \end{tabular} &
\begin{tabular}{c} \textbf{Avg.}\\ {\scriptsize Rank} \end{tabular} \\
\midrule
ERM & 5.4 (1.0) & 18.2 (1.3) & 40.4 (3.8) & 27.8 (5.0) & 8.0 \\
GIN & 6.0 (0.7) & 13.4 (1.1) & \underline{$44.9 (4.0)$} & 20.3(4.3) & 6.5 \\
GINIPA  & $\underline{37.8 (1.5)}$ & $23.8$ (1.7) & $44.5$(3.8) & $27.9$(5.1) & 3.5 \\
MixUp    & $15.6$ (1.6) & $\underline{33.7 (1.2)}$ & $43.7$ (3.9) & $\underline{53.2 (3.4)}$ & 3.5 \\
CutOut & 8.4 (1.1) & 18.2 (1.4) & 41.3 (3.6) & 40.8 (4.5) & 5.5 \\
DSBN     & $19.3$ (1.8) & $22.2$ (1.5) & $44.2$ (3.7) & $44.1$ (5.2) & 3.5 \\
MixStyle & 3.8 (0.8) & 17.6 (1.1) & 37.7 (3.5) & 25.0 (5.1) & 8.3 \\
RSC & 3.6 (1.0) & 18.3 (1.3) & 44.1 (3.9) & 28.5 (5.2) & 6.0 \\
\midrule
Ours ($r=0.5$)     & $\textbf{48.1 (1.7)}$ & $\textbf{40.3 (1.5)}$ & $\textbf{46.3 (3.8)}$ & $\textbf{64.2 (4.4)}$ & $\textbf{1.0}$ \\
\bottomrule
\end{tabular}
\label{tab:main_results_fewshot_r0.5}
\end{table}

\subsection{Robustness Analysis.} 
\label{app:robustness_analysis}
We now assess how robust networks trained with each DG method are to domain shifts that can be simulated and are commonly encountered in practice. As domain shifts/corruptions, we employ randomly generated bias fields and gamma/contrast adjustments as corruptions, as these frequently change between MRI domains. As random bias fields vary spatially, we generate three different bias fields at each corruption level and report the average performance, whereas gamma corruption is deterministic. We note that all methods were trained with strong gamma and bias augmentation, except for nnUNet, which uses its default (conservative) augmentation pipeline and is thus excluded from this experiment for fair comparison. As shown in Figure~\ref{fig:corruptions}, \methodname maintains its robustness under gamma and bias corruptions across most datasets on held-out test datasets. We again find that the next most stable method changes in between datasets, indicating that \methodname offers the most consistent benefits across anatomies and domain shift types.

We now also examine the out-of-distribution robustness of the converged optimal weights of each DG method on held-out test data. As flat minima often generalize better~\cite{hochreiter1997flat, keskar2016large, li2018visualizing}, methods that maintain performance under weight perturbations are more robust. To assess flatness, we add Gaussian noise of varying magnitudes to all parameters of the trained predictors and measure the resulting performance drop. For each noise level~$\alpha$, we average results over five random perturbations. Figure~\ref{fig:corruptions} shows that \methodname degrades gracefully under increasing noise, indicating a more robust solution than competing methods, which decline sharply with weight corruption across most datasets.

\begin{figure}[!t]
    \centering
    \includegraphics[width=1.0\linewidth]{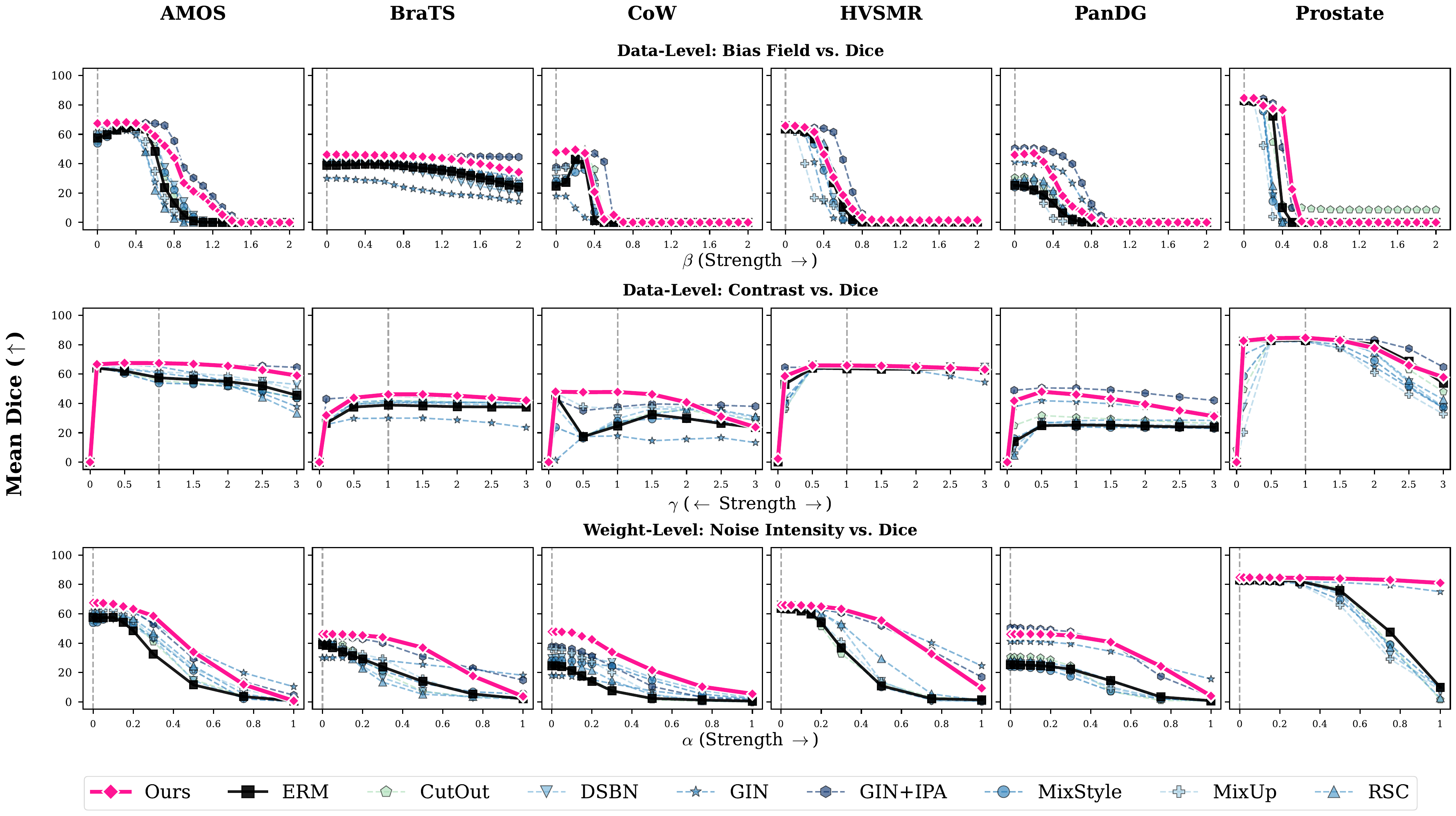}
    \caption{\textbf{Corruption Analysis.} We apply simulated domain shifts with increasing strengths, via bias (\textbf{top}) and contrast/gamma (\textbf{middle}) corruptions and observe that \methodname maintains high robustness. In the \textbf{bottom} row, we perturb the trained model weights with additive Gaussian noise, and find that \methodname is stable under these corruptions, indicating a flatter and more generalizable solution. Dashed vertical line in all plots signifies the "identity" or no corruption level.}
    \label{fig:corruptions}
\end{figure}

\subsection{Runtime Analysis.}
\label{app:runtime}
\methodname adds a forward pass through a frozen foundation model $f_{\phi}$ whose computational overhead should be quantified. Here, we calculate the training time for 1000 training steps for our method, GIN+IPA~\cite{ouyang2022causality}, MixUp~\cite{zhang2017mixup}, and ERM, on the BraTS benchmark and average this across 3 different trials. All training and inference runtimes below were calculated on a single NVIDIA RTX 6000 Ada.

As indicated by the results in Table~\ref{tab:runtime}, we can see that our method introduces a 25.8\% overhead during training when compared to ERM. However, in the setting of high-performance domain generalization methods, we posit that this overhead is not large. For example, the most competitive method, GIN+IPA, introduces a 449.7\% increased train time. This is partly due to its B-spline computation and increased effective batch size (see Section~\ref{app:baseline_implementation} for details). 

To quantify the overhead that \methodname adds at inference, we compare runtimes across methods using 100 samples from BraTS and average these times across 3 different trials. As reported in Table~\ref{tab:inference_runtime}, our method has causes a 8.0\% (17.5ms) increase in inference time, which is well within limits for real-time volumetric analysis (< 0.5s). We consider this computational overhead a favorable trade-off given the out-of-distribution robustness gains reported in Section~\ref{subsec:main_results}.  We speculate that kernel-level and operator-fusing optimization methods would further close this gap.

\begin{table}[!ht]
\centering
\begin{minipage}[t]{0.48\textwidth}
\centering
\caption{\textbf{Training Runtime.} Average runtime over 3 trials on BraTS for 1000 training steps. $\Delta$ shows relative increase vs. ERM.}
\label{tab:runtime}
\footnotesize
\setlength{\tabcolsep}{3pt}
\begin{tabular}{lll}
\toprule
Method & Run Time (s) & $\Delta$ (\%) \\
\midrule
ERM & 259.3 $\pm$ 0.6 & 0\% \\
Ours & 326.1 $\pm$ 0.3 & +25.8\% \\
GIN & 269.6 $\pm$ 0.3 & +4.0\% \\
GIN+IPA & 1165.8 $\pm$ 0.5 & +449.7\% \\
MixUp & 262.1 $\pm$ 0.4 & +1.1\% \\
\bottomrule
\end{tabular}
\end{minipage}
\hfill
\begin{minipage}[t]{0.48\textwidth}
\centering
\captionof{table}{\textbf{Inference Runtime.} Average inference time over 3 trials on BraTS for 20 volumes. $\Delta$ shows relative increase vs. ERM.}
\label{tab:inference_runtime}
\vspace{6pt}
\footnotesize
\setlength{\tabcolsep}{3pt}
\begin{tabular}{lll}
\toprule
Method & Time (ms/vol) & $\Delta$ (\%) \\
\midrule
ERM & 217.6 $\pm$ 2.3 & 0\% \\
Ours & 235.1 $\pm$ 12.8 & +8.0\% \\
\bottomrule
\end{tabular}
\end{minipage}
 
\end{table}

\subsection{Broader Impacts}

This work pertains to achieving better out-of-domain robustness in biomedical segmentation tasks. There are potential societal benefits to this work (\eg, reducing clinician annotation effort for training these networks) and it is implausible that negative impacts will arise from it.

\clearpage
\newpage

\section{Implementation Details}
\label{app:implementation}

\subsection{Training Setup}
We use the representations from its anatomix's~\cite{dey2024learning} penultimate convolutional block and instance normalize its representations to be channel-wise zero mean and unit standard deviation. For training, all methods use AdamW~\cite{loshchilov2017decoupled} with weight decay of $10^{-4}$
, a per-step cosine learning rate schedule from $2\times 10^{-4}$ to $0$, and a Dice and cross entropy loss (equally weighted).

\subsection{Augmentations}
To best estimate the effects of \methodname in addition to data augmentation, we use an extensive augmentation pipeline (provided below) for all the experiments, except for nnUNet~\cite{isensee2021nnu}, which has its own augmentation framework and is recommended by its authors to be used as-is. Listed functions correspond to MONAI conventions~\cite{cardoso2022monai}. We use flips, rotations, resolution degradations, noise, bias fields, contrast adjustments, blurring, sharpening, and affine deformations, each with 0.33 probability. We omit flips and rotations in datasets where left/right correspondence is important, as indicated below.

\begin{minted}[fontsize=\scriptsize, breaklines, frame=lines, bgcolor=gray!5]{python}
opts.aug_prob = 0.33
train_transforms = Compose([
    LoadImaged(keys=["image", "label"]),
    EnsureChannelFirstd(keys=["image", "label"]),
    EnsureTyped(keys=["image", "label"]),
    Spacingd(keys=["image", "label"], pixdim=opts.spacing, mode=("bilinear", "nearest")),
    SpatialPadd(keys=["image", "label"], spatial_size=opts.crop_size),
    ScaleIntensityd(keys="image") if opts.dataset not in ["AMOS", "COW"] else ScaleIntensityRanged(keys="image", a_min=-1024, a_max=1024, b_min=0.0, b_max=1.0,clip=True,),
    RandAxisFlipd(keys=["image", "label"], prob=opts.aug_prob) if opts.dataset not in ["AMOS", "COW"] else Identity(keys=["image", "label"]),
    RandRotate90d(keys=["image", "label"], prob=opts.aug_prob, max_k=3, spatial_axes=(0, 1)) if opts.dataset not in ["AMOS", "COW"] else Identity(keys=["image", "label"]),
    RandSimulateLowResolutiond(keys="image", prob=opts.aug_prob, zoom_range=(0.25, 1.0),),
    RandGaussianNoised(keys="image", prob=opts.aug_prob),
    RandBiasFieldd(keys="image", prob=opts.aug_prob, coeff_range=(0.0, 0.1)),
    RandGibbsNoised(keys="image", prob=opts.aug_prob, alpha=(0.0, 0.33)),
    RandAdjustContrastd(keys="image", prob=opts.aug_prob),
    RandGaussianSmoothd(keys="image", prob=opts.aug_prob, sigma_x=(0.0, 0.1), sigma_y=(0.0, 0.1), sigma_z=(0.0, 0.1)),
    RandGaussianSharpend(keys="image", prob=opts.aug_prob),
    RandAffined(keys=["image", "label"], prob=opts.aug_prob, mode=("bilinear", "nearest"), rotate_range=(np.pi / 4, np.pi / 4, np.pi / 4), scale_range=(0.2, 0.2, 0.2), shear_range=(0.2, 0.2, 0.2), spatial_size=opts.crop_size, padding_mode='zeros',),
    ScaleIntensityd(keys="image"), [0,1]
])
\end{minted}

\begin{table}[th!]
\centering
\caption{U-Net architectural details. \texttt{Conv-BN-ReLU} refers to a sequence of 3D convolutional with $3 
\times 3 \times 3$ kernels, batch normalization, and pointwise ReLU non-linearities. $n_{in}$ is the number if input channels. $n_c$ is the channel width multiplier and $n$ is the number of output channels. }
\begin{tabular}{c l}
\hline
\textbf{Block} & \textbf{Contents} \\
\hline
0  & Conv-BN-ReLU($n_{in}, n_c$) \\
1  & Conv-BN-ReLU($n_c$) \\
2  & Conv-BN-ReLU($n_c$) \\
3  & MaxPool(2), Conv-BN-ReLU($2n_c$) \\
4  & Conv-BN-ReLU($2n_c$) \\
5  & Conv-BN-ReLU($2n_c$) \\
6  & MaxPool(2), Conv-BN-ReLU($4n_c$) \\
7  & Conv-BN-ReLU($4n_c$) \\
8  & Conv-BN-ReLU($4n_c$) \\
9  & MaxPool(2), Conv-BN-ReLU($8n_c$) \\
10  & Conv-BN-ReLU($8n_c$) \\
11  & Conv-BN-ReLU($8n_c$) \\
12  & MaxPool(2), Conv-BN-ReLU($16n_c$) \\
13 & Conv-BN-ReLU($16n_c$) \\
14 & Conv-BN-ReLU($16n_c$) \\
15 & MaxPool(2), Conv-BN-ReLU($32n_c$) \\
16 & Conv-BN-ReLU($32n_c$) \\
17 & Conv-BN-ReLU($32n_c$) \\
18 & Upsample $2\times$, Concatenate with layer 14 \\
19 & Conv-BN-ReLU($32n_c$) \\
20 & Conv-BN-ReLU($32n_c$) \\
21 & Upsample $2\times$, Concatenate with layer 11 \\
22 & Conv-BN-ReLU($16n_c$) \\
23 & Conv-BN-ReLU($16nc$) \\
24 & Upsample $2\times$, Concatenate with layer 8 \\
25 & Conv-BN-ReLU($4nc$) \\
26 & Conv-BN-ReLU($4nc$) \\
27 & Upsample $2\times$, Concatenate with layer 5 \\
28 & Conv-BN-ReLU($2nc$) \\
29 & Conv-BN-ReLU($2nc$) \\
30 & Upsample $2\times$, Concatenate with layer 2 \\
31 & Conv-BN-ReLU($nc$) \\
32 & Conv-BN-ReLU($nc$) \\
33 & Conv-BN-ReLU($nout$) \\
\hline
\end{tabular}
\label{tab:network-architecture}
\end{table}

\subsection{Architecture}
We provide the base 3D UNet segmentation architecture used throughout this paper in Table~\ref{tab:network-architecture}. All methods used this base architecture except for nnUNet~\cite{isensee2021nnu}, which is recommended to be used as is.

\subpara{Notation:} 
\begin{itemize}
    \item $nc$: base number of channels
    \item $nout$: number of output channels in the final layer
    \item Conv-BN-ReLU$(c)$: A Convolution $\rightarrow$ BatchNorm $\rightarrow$ ReLU with $c$ output channels block
    \item MaxPool$(k)$: Max pooling with kernel size $k$
    \item Upsample $k\times$: nearest-neighbor upsampling by factor $k$
    \item ``Concatenate with layer X'': skip connection via concatenation
\end{itemize}

\subsection{Feature Extractor}
Our paper evaluates two large-scale medical imaging foundation models as representation extractors for \methodname: anatomix~\cite{dey2024learning} and nnInteractive~\cite{isensee2025nninteractive}. Below, we describe their specific model configurations and checkpoints used in our experiments.

\subpara{anatomix.} We use the publicly released weights provided in the official GitHub repository, specifically the \texttt{anatomix.pth} checkpoint. These weights correspond to the pretrained model described in the referenced work and require no additional modification.

\subpara{nnInteractive.} For nnInteractive, we rely on the reference implementation available in the associated GitHub repository. All experiments use nnInteractive v1.0, fold 0, with the \texttt{checkpoint\_final.pth} model. Consistent with the setup described in the main text, we extract features from the final decoder layer for all comparisons and ablations.

\subsection{Dataset Pre-processing and Splits}
\label{app:dataset_preprocessing}
Here, we describe the domain shifts present in the training, validation, and testing sets and summarize the corresponding dataset statistics and splits in Table~\ref{tab:preprocessing_datasets}. Specifically, the training and evaluation domains for each dataset are noted in Section~\ref{sec:experiments} and do not share any subjects between any of the training, validation, and testing splits. Dataset splitting scripts are provided in the accompanying code for further clarity. 
All methods were trained, validated, and tested using the same voxel spacing.
Voxel spacing was determined using the nnUNet planning pipeline~\cite{isensee2021nnu} run on the train split in isolation, with the exception of AMOS, for which we followed the AMOS-specific recommendations from the nnUNet extensions~\cite{isensee2023extending}. Crop sizes were selected based on memory constraints and to ensure compatibility with the feature extractors used in our experiments. Below, we describe the specific datasets and any other additional information important for training. We refer readers to the referenced work for download links.

\begin{table}[t]
\centering
\caption{\textbf{Dataset details.} Pre-processing details such as crop size and spacing, alongside splits used for each dataset in the experiments.}
\begin{tabular}{l c c c c c}
\toprule
\textbf{Dataset} & \textbf{Crop Size} & \textbf{Spacing (mm)} & \textbf{\# Train} & \textbf{\# Val.} & \textbf{\# Test} \\
\midrule
AMOS & $(192, 192, 128)$ & $(1.0, 1.0, 1.5)$ & $300$ & $30$ & $30$ \\
BraTS & $(128, 128, 128)$ & $(1.0, 1.0, 1.0)$ & $1020$ & $48$ & $96$ \\
CoW & $(192, 160, 128)$ & $(0.3, 0.3, 0.6)$ & $75$ & $20$ & $30$ \\
HVSMR & $(128, 160, 128)$ & $(0.74, 0.74, 0.8)$ & $23$ & $18$ & $19$ \\
PancreasDG & $(256, 224, 64)$ & $(1.09, 1.09, 4.4)$ & $53$ & $100$ & $100$ \\
Prostate & $(224, 256, 64)$ & $(0.52, 0.51, 1.25)$ & $49$ & $25$ & $42$ \\
\bottomrule
\end{tabular}
\label{tab:preprocessing_datasets}
\end{table}

\subsection{Baseline Implementations}
\label{app:baseline_implementation}
We summarize how each baseline method was implemented for the comparisons in the main paper. For reproducibility, we refer readers to the official open-sourced implementations provided by the authors of each respective method. 

\subpara{nnU-Net.}  
We used the latest publicly available version of nnU-Net from the official repository~\cite{isensee2021nnu}. Training plans were generated \textit{exclusively from the training splits} we randomly selected to prevent any leakage from the validation or test sets. For evaluation, we selected the best-performing validation checkpoint and disabled test-time augmentation to ensure a fair comparison across baselines. Further, its default five-network ensembling was not used in order to achieve a one-to-one comparison.

\subpara{GIN and GIN+IPA.}  
GIN consists of a shallow convolutional block that randomly samples convolutional kernels at each forward pass; we used the authors’ \texttt{GINGroupConv3D} module with default settings, as recommended in the reference code~\cite{ouyang2022causality}.  Since no open-source implementation of IPA is provided, we implemented it following the descriptions in the original work. IPA generates a GPU-based B-spline bias field with user-controlled spacing; we used a spacing of \([16,16,16]\) for all experiments. IPA also incorporates a consistency regularization term, for which we followed the parameterization used in the reference code and tuned its consistency weight on the validation set. 

GIN+IPA effectively increases the batch size because multiple perturbations are applied per input (\eg, an initial batch size of 4 becomes an effective batch size of 12). This increase in batch size in addition to the B-spline computation causes a significant training time increase as demonstrated in Table~\ref{tab:runtime}. Additionally, both GIN and GIN+IPA were \textit{highly} sensitive to initialization, necessitating the use of multiple random seeds to avoid divergent training. We trained GIN and GIN+IPA with three different random seeds each, and selected the corresponding models with the best performance on the validation set.

\subpara{MixUp.}  
MixUp~\cite{zhang2017mixup} interpolates pairs of samples within each batch, producing soft labels that encourage smoother decision boundaries. We applied MixUp with a probability $p$ that was tuned on the validation set and sampled interpolation coefficients from the recommended \(\text{Beta}(0.4, 0.4)\) distribution.

\subpara{CutOut.}  
CutOut~\cite{devries2017improved} masks out randomly positioned cubic regions from both the image and its label, reducing the model’s reliance on small, highly discriminative regions. The size of the masked cube was set to \(0.33 \times\) the crop size along each dimension (rounded to the nearest integer). CutOut was applied with a probability $p$ tuned on the validation set. Because the input was scaled to \([0,1]\), we zero-centered the volume before applying CutOut and then shifted it back to \([0,1]\), producing a neutral ``gray'' cube instead of a zero-valued region. This stabilized training when using batch normalization.

\subpara{DSBN.}  
Domain-Specific Batch Normalization~\cite{zhou2022generalizable} replaces standard batch normalization with multiple domain-specific branches. Since no domain labels are provided, domains were assigned based on whether the Bézier intensity transformation applied to the foreground increased or decreased the intensity. Transformations were applied with a probability $p$ tuned on the validation set. %

\subpara{MixStyle.}  
MixStyle~\cite{zhou2021domain} mixes feature statistics within each batch to simulate domain variation. Because explicit domain labels are not available, we mixed statistics randomly between samples in the batch. We applied MixStyle to the UNet encoder. We applied the MixStyle with probability $p$ tuned on the validation set.

\subpara{RSC.}  
RSC~\cite{huang2020self} requires two forward passes: one to compute standard gradients and another in which the top 20\% of dominant gradients are masked. This discourages the model from over-relying on a small set of highly discriminative features. As high-level features emerge in deeper encoder layers, we applied RSC at the deepest encoder block. We applied RSC with probability $p$ tuned on the validation set.

\subpara{Stable Feature Boosting (SFB).}
SFB~\cite{eastwood2023spuriosity} approaches the domain generalization problem by training two separate model: unstable and stable. At test time, SFB uses a calibrated stable model to produce pseudo-labels. These pseudo-labels serve to supervise the unstable model during test time. A bias correction procedure is used to mitigate the wrong predictions made by the stable feature during test time. In our context, the "stable" model is a model trained on the stable representations only and the "unstable" model is a model trained on the unstable representations only. We first calibrate the stable model via temperature scaling on the validation set. For each test volume, the stable model products calibrated softmax probabilities. The unstable model is then adapted over 3 outer rounds: in each round, it is finetuned for $S$ gradients steps using a KL divergence loss against the pseudo-labels with a learning rate of $LR$. Unstable model's predictions are bias corrected. Final segmentation is the argmax of the pseudo-labels after all 3 rounds. $S$ and $LR$ are tuned on the validation set.

\subpara{Alpha Tuning.}
We interpreted Alpha Tuning~\cite{anti-causal} as a linear combination between a stable model and a unstable model. Specifically, we fused the logits according to $\text{final\_logits} = (1 - \alpha) \cdot h^{\text{stable}}_{\theta} + \alpha  \cdot h^{\text{unstable}}_{\theta}$ where $\alpha$ is tuned on the validation set.

\subpara{Entropy Minimization (Ent. Min.).}
TENT~\cite{wang2020tent} adapts a pretrained model at test time by minimizing the entropy of its predictions on each test volume. Only the batch normalization (BN) parameters are updated, while all other weights are frozen. For each test volume, we reset the model to its pretrained weights and run $S$ gradient steps of SGD with learning rate of $LR$, where each step is a random spatial crop. $S$ and $LR$ are tuned on the validation set.

\end{document}